\documentclass[11pt]{article}

\PassOptionsToPackage{numbers, compress}{natbib}
\usepackage{amsmath,amssymb,amsthm,fullpage,mathrsfs,pgf,tikz,caption,subcaption,mathtools}
\usepackage[ruled,vlined,linesnumbered]{algorithm2e}
\usepackage{natbib}
\usepackage[utf8]{inputenc} 
\usepackage[T1]{fontenc}    
\usepackage{hyperref}       
\usepackage{url}            
\usepackage{booktabs}       
\usepackage{amsfonts}       
\usepackage{nicefrac}       
\usepackage{microtype}      
\usepackage{times}
\usepackage[margin=1in]{geometry}
\usepackage{bbm}
\usepackage{enumitem}
\usepackage{xcolor}
\usepackage{cleveref}
\usepackage{bm}

\newtheorem{theorem}{Theorem}[section]
\newtheorem{corollary}[theorem]{Corollary}
\newtheorem{proposition}[theorem]{Proposition}
\newtheorem{lemma}[theorem]{Lemma}
\newtheorem{definition}[theorem]{Definition}

\newtheorem{observation}[theorem]{Observation}
\newtheorem{claim}[theorem]{Claim}

\newcommand{\R}{\mathbb{R}}

\DeclareMathOperator{\poly}{poly}

\newcommand\norm[1]{\left\lVert#1\right\rVert}

\author{Max Hopkins\thanks{Department of Computer Science, Princeton, NJ 08544. Email: \texttt{mh4067@princeton.edu}.}
}
\title{Hypercontractivity on HDX II: Symmetrization and $q$-Norms}
\date{}
\begin{document}
\maketitle
\begin{abstract}
    Bourgain's symmetrization theorem is a powerful technique reducing boolean analysis on product spaces to the cube. It states that for any product $\Omega_i^{\otimes d}$, function $f: \Omega_i^{\otimes d} \to \R$, and $q > 1$:
    \[
    \norm{T_{\frac{1}{2}}f(x)}_q \leq \norm{\tilde{f}(r,x)}_{q} \leq \norm{T_{c_q}f(x)}_q
    \]
    where $T_{\rho}f = \sum\limits \rho^Sf^{=S}$ is the noise operator and $\tilde{f}(r,x) = \sum\limits r_Sf^{=S}(x)$ `symmetrizes' $f$ by convolving its Fourier components $\{f^{=S}\}_{S \subseteq [d]}$ with a random boolean string $r \in \{\pm 1\}^d$.

    In this work, we extend the symmetrization theorem to high dimensional expanders (HDX). Building on (O'Donnell and Zhao 2021), we show this implies nearly-sharp $(2{\to}q)$-hypercontractivity for partite HDX. This resolves the main open question of (Gur, Lifshitz, and Liu STOC 2022) and gives the first fully hypercontractive subsets $X \subset [n]^d$ of support $n\cdot\exp(\poly(d))$, an exponential improvement over Bafna, Hopkins, Kaufman, and Lovett's $n\cdot\exp(\exp(d))$ bound (BHKL STOC 2022). Adapting (Bourgain JAMS 1999), we also give the first booster theorem for HDX, resolving a main open question of BHKL.

    Our proof is based on two elementary new ideas in the theory of high dimensional expansion. First we introduce `$q$-norm HDX', generalizing standard spectral notions to higher moments, and observe every spectral HDX is a $q$-norm HDX. Second, we introduce a simple method of coordinate-wise analysis on HDX which breaks high dimensional random walks into coordinate-wise components and allows each component to be analyzed as a \textit{$1$-dimensional} operator locally within $X$. This allows for application of standard tricks such as the replacement method, greatly simplifying prior analytic techniques.
\end{abstract}
\section{Introduction}

Recent years have seen the development of a powerful theory of boolean function analysis \textit{beyond the cube}, with breakthrough applications spanning hardness of approximation \cite{dinur2017exponentially,dinur2018towards,subhash2018pseudorandom,kaufman2022improved,minzer2024near}, extremal and probabilistic combinatorics \cite{friedgut1999sharp,keevash2021global,ellis2024product}, approximate sampling \cite{lee1998logarithmic,chen2021optimal,jain2023optimal}, and even quantum communication \cite{arunachalam2023one}. Chief among analytic tools (and core to all above applications) is the theory of \textit{hypercontractivity} \cite{bonami1970etude}, a powerful functional inequality controlling the `smoothness' of low-degree functions.
Popularized in the seminal work of Kahn, Kalai, and Linial \cite{kahn1988influence}, variants of the hypercontractive inequality have since been established in a variety of settings including products \cite{bourgain1992influence,talagrand1994russo,friedgut1996every,friedgut1999sharp,zhao2021generalizations,keevash2021global,keevash2023sharp}, the slice \cite{khot2018small}, the shortcode \cite{barak2012making}, the Grassmann \cite{subhash2018pseudorandom,ellis2023analogue}, and various groups \cite{filmus2020hypercontractivity,ellis2024product}.

Many potential applications of hypercontractivity in complexity (e.g.\ toward low-soundness PCPs \cite{dinur2017exponentially}, sparsest cut \cite{kane2013prg}, the unique games conjecture \cite{barak2018small}) are hindered by the fact that almost all known variants hold for \textit{dense} domains. Indeed it wasn't until 2021 and the works of Gur, Lifshitz, and Liu \cite{gur2021hypercontractivity} and Bafna, Hopkins, Kaufman, and Lovett \cite{bafna2021hypercontractivity} that basic hypercontractivity was shown for any sparse domain at all. These works gave a preliminary set of hypercontractive inequalities for sparse yet strongly `pseudorandom' subsets of $[n]^d$ called \textit{high dimensional expanders} (HDX), a remarkable family of objects with breakthrough applications in closely related areas including the Mihail-Vazirani conjecture \cite{kaufman2020high,anari2019log}, c3-LTCs \cite{dinur2022locally,panteleev2022asymptotically}, and quasilinear size PCPs \cite{bafna2024quasi}.

While the first of their kind, \cite{gur2021hypercontractivity} and \cite{bafna2021hypercontractivity}'s hypercontractive bounds fall short of their classical counterpart in several critical aspects. In \cite{gur2021hypercontractivity}, the authors' bound suffers asymptotically sub-optimal dependence on degree, ruling out many traditional applications.\footnote{In particular, \cite{gur2021hypercontractivity} prove a simplified variant of hypercontractivity called the `Bonami lemma' but with super-exponential dependence on degree. As a result, one cannot recover standard hypercontractivity of the noise operator from their work.} Meanwhile \cite{bafna2021hypercontractivity} prove a variant with the correct asymptotic dependence, but only for quite ``dense'' HDX (i.e., families $X \subseteq [n]^d$ of size $|X| \approx n\cdot 2^{2^{O(d)}}$), a severe limitation since many applications require taking $d$ large. As such, \cite{gur2021hypercontractivity} and \cite{bafna2021hypercontractivity} left open an obvious question: \textit{can we achieve the best of both worlds?} Even further, are there sparse subsets matching \textit{sharp} hypercontractive inequalities for $[n]^d$? I.e.\ stronger $(2{\to}q)$-type bounds needed for recent breakthrough applications in combinatorics and number theory \cite{keller2023sharp,green2024improved}?


In this work, we resolve these questions up to a log factor by taking a new approach to hypercontractivity on HDX: \textit{symmetrization}. Symmetrization is a classical technique of Kahane \cite{kahane1985some} and Bourgain \cite{bourgain1979walsh} which analyzes high moments of functions on product spaces $f: \Omega_i^{\otimes d} \to \R$ by convolving them with a random boolean string $r \in \{\pm 1\}^d$. Up to the application of a small amount of noise, Bourgain proved that the resulting function $\tilde{f}:\{\pm 1\}^d \times \Omega_i^{\otimes n} \to \R$ has the same behavior of $f$, allowing one to analyze $\tilde{f}$ instead using standard analysis on its boolean component. We prove Bourgain's theorem holds for high dimensional expanders up to a $(1 \pm o(1))$ factor. 

Building on \cite{zhao2021generalizations}, we then use symmetrization to prove a nearly-sharp $(2{\to}q)$-hypercontractive inequality and a booster theorem on HDX, resolving the main open questions of \cite{bafna2021hypercontractivity,gur2021hypercontractivity}. Combined with \cite{lubotzky2005explicit}, our bounds give the first fully hypercontractive sparsification $X \subset [n]^d$ of support $n \cdot 2^{\text{poly}(d)}$, coming much closer to the size of heavily used `weakly pseudorandom' sparsifications of $[n]^d$ like expander walks. This reduced support is critical since it corresponds to the `blowup' incurred using $X$ in application, e.g.\ as the rate blowup in applications to codes \cite{dinur2019list,alev2020list,dikstein2024chernoff} or size blowup in applications to PCPs and fault-tolerant networks \cite{bafna2024quasi,bafna2024bconstant}. In the latter vein, it is worth highlighting that (global) hypercontractivity has classically played an important role in PCPs (see, e.g.\ \cite{haastad2001some,dinur2018towards}). In fact, \cite{bafna2024quasi}'s PCP construction using HDX hit a barrier at low constant soundness in part due to lack of an optimal agreement test, variants of which are known to follow from hypercontractivity \cite{dinur2017exponentially,barak2018small,kaufman2022improved,minzer2024near}. Combined with our results, extending the latter connection to HDX could lead to substantially improved PCPs in the notoriously challenging \textit{subconstant} soundness regime.\footnote{We remark it is known that hypercontractivity alone is not sufficient for agreement testing on HDX \cite{dikstein2024agreement,bafna2024characterizing}, but it is plausible optimal testers follow from combining hypercontractivity with the topological conditions from these works.}

Our proof of the symmetrization theorem is based on two new ideas in the theory of high dimensional expansion. First, we introduce the notion of \textit{$q$-norm HDX}, which bounds the local divergence of a hypergraph from a product space in $q$-norm, and show any strong enough spectral HDX is a $q$-norm HDX. This allows us to manipulate higher norms directly and avoid error terms that hampered previous techniques. Second, we introduce a basic variant of coordinate-wise analysis for the noise operator $T_r$ on $q$-norm HDX, showing it can be approximately broken into coordinate-wise components $T_{r_i}^i$ that correctly localize to the standard $1$-dimensional operators $T_{r_i}$ when restricted to the $i$th coordinate. This allows us to apply the replacement method, greatly simplifying \cite{gur2021hypercontractivity,bafna2021hypercontractivity} and allowing us to achieve near-sharp bounds downstream.

\subsection{Background}
\paragraph{High Dimensional Expanders:} We study (weighted) \textit{partite hypergraphs} $X \subset [n]^d$, which we view as $d$-dimensional distributions $X=(X_1,\ldots,X_d)$. $X$ is called a \textit{$\gamma$-product}\footnote{$\gamma$-products, introduced formally in \cite{gur2021hypercontractivity}, are essentially equivalent to the standard notion of spectral HDX \cite{dikstein2019agreement}.} if for every $i \neq j$
\begin{enumerate}
    \item The bipartite graph corresponding to the marginal $(X_i,X_j)$ is a $\gamma$-spectral expander
    \item This holds even conditioning on any (feasible) values of $X_S$ for any $i,j \not\in S \subseteq [d]$.
\end{enumerate}
Note $[n]^d$ itself is a $0$-product of support $n^d$ and that there exist (strongly explicit) constructions of $\gamma$-products of support $n\cdot \gamma^{-O(d^2)}$ \cite{lubotzky2005explicit}. We remark all results in this paper hold in the setting of general domains $X \subset \Omega_i^{\otimes d}$; we focus on $[n]^d$ for simplicity of presentation.
\paragraph{Basic Analysis on Hypergraphs}
Denote the space of functions on $X$ as $C_d \coloneqq \{f:X \to \R\}$. Assuming the weights of $X$, denoted $\Pi(x)$, form a distribution, they induce a natural expectation operator and $L_p$-space over $C_d$ defined by $\mathbb{E}[f] = \sum\limits_{x \in X} \Pi(x)f(x)$ and $\norm{f}_p = \mathbb{E}[|f|^p]^{1/p}$.
Given a linear operator $T:C_d \to C_d$, we define its $p$-norm in the standard way as $\max_{f \in C_d}\left\{\frac{\norm{Tf}_p}{\norm{f}_p}\right\}$. 

It will often be useful for us to analyze $f \in C_d$ inside \textit{restrictions} of $X$. Formally, given a sub-assignment $x_S$ to the marginal $X_S$, define $f|_{x_S}(x_{\bar{S}})=f(x)$ to be the localization of $f$ to the restricted domain $X_{\bar{S}}$ \textit{conditioned on fixing $X_S=x_S$}. This domain is called the link and denoted by $X_{x_S}$ (informally, the link just fixes the $S$-coordinates to value $x_S$, and we think of $f|_{x_S}$ as a function of the remaining coordinates).

\paragraph{Noise Operator and Total Influence:}

The \textit{noise operator} $T_\rho: C_d \to C_d$ is a classical smoothing procedure in boolean analysis which given a string $x=(x_1,\ldots,x_d) \in X$, averages over $y$ generated by re-sampling each coordinate of $x$ independently with probability $1-\rho$. The \textit{total influence} of a $\{\pm 1\}$-valued function $f \in C_d$ measures its average sensitivity to re-sampling each \textit{individual} coordinate
\[
\mathbf{I}[f] = \sum\limits_{i=1}^d \mathbb{E}_x[\Pr[f(x) \neq f(x^{(i)})]],
\]
where $x^{(i)}$ is generated by re-sampling the $i$th coordinate of $x$.

\paragraph{Fourier Analysis on Hypergraphs:}
The classical theory of Fourier analysis on products is given by the \textit{Efron-Stein Decomposition}, whose components $\{f^{=S}\}_{S \subseteq [d]}$ are
\[
f^{=S}(y) = \sum\limits_{T \subseteq S}(-1)^{|S \setminus T|}\mathbb{E}[f(x)~|~x_T=y_T].
\]
It is easy to check that $f=\sum\limits_{S \subseteq [d]}f^{=S}$. We write $f^{\leq i}=\sum\limits_{|S| \leq i} f^{=S}$ as the degree-$i$ component of $f$. On product spaces, $\{f^{=S}\}$ is an orthogonal eigenbasis for $T_\rho$:
\[
T_\rho f = \sum\limits_{S \subseteq [d]}\rho^{|S|}f^{=S}.
\]
This leads to a standard generalization of the noise operator to vectors $r \in \R^d$
\[
T_r f \coloneqq \sum\limits_{S \subseteq [d]} r_Sf^{=S},
\]
where $r_S = \prod\limits_{i \in S} r_i $. Finally, the \textit{symmetrization} of $f$, denoted $\tilde{f}: \{\pm 1\}^d \times X \to \R$ is the function
\[
\widetilde{f}(r,x) = T_rf(x) = \sum\limits_{S \subseteq [d]} r_Sf^{=S}(x).
\]
Note $\tilde{f}$ is the unique function whose $X$-restrictions $\tilde{f}|_x(r)=\tilde{f}(x,r)$ have Fourier coefficients $f^{=S}(x)$.
\subsection{Results}\label{sec:intro-results}

Our first main contribution is the following extension of Bourgain's symmetrization theorem to HDX:
\begin{theorem}[The Symmetrization Theorem (Informal \Cref{cor:sym-products})]\label{thm:symmetrization-intro}
    Let $q>1$ and $X$ be a $d$-partite $\gamma$-product satisfying $\gamma \leq 2^{-\Omega_q(d)}$. Then for any $f:X \to \R$
    \[
    (1-o_{\gamma}(1))\norm{\widetilde{T_{c_q}f}}_q \leq \norm{f}_q \leq (1+o_{\gamma}(1))\norm{\widetilde{T_{2}f}}_q.
    \]
    for some constant $c_q>0$ depending only on $q$.
\end{theorem}
The symmetrization theorem is a powerful tool for extending classical structure theorems on the cube to product spaces. In this work, we focus mainly on its application to the \textit{hypercontractive inequalities}. On the boolean cube, the classical $(2{\to}q)$-hypercontractive theorem (or rather its equivalent form as the `Bonami Lemma') states the degree-$i$ part of any function $f: \{0,1\}^d \to \R$ is `smooth' in the sense that for any $q > 2$
\[
\norm{f^{\leq i}}_q \leq \sqrt{q-1}^i\norm{f^{\leq i}}_2.
\]
This statement has incredibly far reaching consequences in analysis, combinatorics, and theoretical computer science, including the KKL Theorem \cite{kahn1988influence} and its myriad applications.

Unfortunately, the Bonami lemma is false on unbalanced product spaces. The obvious counter-examples are dictators, or more generally `local' functions that are dense in some restriction. Consider, for instance, the $i$th dictator on the $p$-biased cube. While $\mathbf{1}_i$ is a degree-$1$ function, it is easy to compute:
\[
\mathbb{E}[\mathbf{1}_i^4] = p \gg \mathbb{E}[\mathbf{1}_i^2]^2 = p^2.
\]
This gives rise to the theory of \textit{global} hypercontractivity that aims to show such local functions are the only counter-examples. Variants of global hypercontractivity have been known for product spaces since the 90s where they lead to the theory of sharp thresholds \cite{friedgut1999sharp}, but tight bounds were only achieved after a long line of work culminating in the recent work of \cite{keller2023sharp} (see \Cref{sec:intro-related-work} for further references).

With this in mind, our second main contribution is a proof of global hypercontractivity for HDX via symmetrization that matches sharp bounds for products up to a log factor.
\begin{theorem}[Global Hypercontractivity (Bonami Form, \Cref{thm:Bonami})]\label{thm:Bonami-intro}
    Let $q \geq 2$ be an even integer and $X$ any $d$-partite $\gamma$-product with $\gamma \leq 2^{-\tilde{\Omega}(qd)}$. Then for any $f:X \to \R$ and $i \leq d$:
    \[
    \norm{f^{\leq i}}_q \leq (500q)^{i}\norm{f}_2^{2/q}\max_{|S| \leq i, x_S}\{\norm{f|_{x_S}}_2^{1-\frac{2}{q}}\}.
    \]
\end{theorem}
To interpret the above, note so long as $f$ is not too dense in any low-dimensional restriction (i.e.\ $f$ is `global'), \Cref{thm:Bonami-intro} recovers the $(2{\to}q)$-Bonami lemma up to constants in the exponent. Indeed on balanced complexes such as the cube, \textit{all} functions are global, so we recover standard $(2{\to}q)$-hypercontractivity.\footnote{In fact even on unbalanced complexes the same argument recovers `general hypercontractivity' which scales with the worst marginal probability, see \cite[Chapter 10]{o2014analysis}.}

On a related note, one might reasonably observe \Cref{thm:Bonami-intro}'s dependence on $q$ and $i$ is not as good as the standard Bonami lemma, degrading from $\sqrt{q}^i$ to $q^i$. As we've alluded to above, this is necessary. \cite{keller2023sharp} recently showed the optimal dependence on product spaces to be $\Theta(\frac{q}{\log(q)})^i$. We remark that while \Cref{thm:Bonami-intro} is worse by a log factor, even the proof of the product case (a very simple extension of \cite{zhao2021generalizations}'s method for $(2{\to}4)$-norm) may be of independent interest, since the only prior known way to achieve linear dependence on $q$ was via \cite{keller2023sharp}'s much more involved reduction to Gaussian hypercontractivity.

Achieving sharp dependence on $q$ and $i$ is often crucial in application (see \cite{keller2023sharp} and followup works \cite{keevash2023sharp,lifshitz2023bounds}) because this coefficient needs to `match' the corresponding eigenvalues of the noise operator $T_\rho$ --- in other words the smaller the coefficient, the less noise required in application. To formalize this connection, we show how to use \Cref{thm:Bonami-intro} to prove a global variant of the classical hypercontractive inequality for the noise operator. Following \cite{keller2023sharp}, we'll call a function $f$ $r$-global if 
\[
\forall S \subset [d], x_S \in X_S: \norm{f|_{x_S}}_2^2 \leq r^{|S|}\norm{f}_2^2.
\]
\Cref{thm:Bonami-intro} then leads to the following near-sharp $(2{\to}q)$-hypercontractive bound for $T_\rho$:
\begin{corollary}[Global Hypercontractivity (Operator Form, Informal \Cref{cor:operator-form})]\label{cor:intro-operator-form}
    Let $r \geq 1$, $q \geq 2$ be an even integer, and $X$ be any $d$-partite $\gamma$-product with $\gamma \leq 2^{-\tilde{\Omega}(qd)}$. Then for any $\rho \leq O(\frac{1}{rq})$ and $r$-global function $f:X \to \R$:
    \[
    \norm{T_\rho f}_q \leq (1+o_\gamma(1))\norm{f}_2.
    \]
\end{corollary}
We emphasize that one cannot recover even a weak form of \Cref{cor:intro-operator-form} via a Bonami lemma with asymptotically sub-optimal dependence on degree (e.g.\ as in \cite{gur2021hypercontractivity}). Here we achieve a near-optimal noise rate of $O(\frac{1}{rq})$, within a log factor of the optimal $\Theta(\frac{\log(q)}{rq})$ achieved in \cite{keller2023sharp} for products.

Before moving to applications, it is worth taking a step back to view these results in a broader context. In particular, it is instructive to view complexes $X$ satisfying such hypercontractive bounds as a useful class of `strongly-pseudorandom' subsets of $[n]^d$. Even restricted to $q=4$, prior to our work the sparsest known subsets had support $n\cdot2^{2^{O(d)}}$ \cite{bafna2021hypercontractivity}. \Cref{thm:Bonami-intro} and \Cref{cor:intro-operator-form} hold for strongly explicit subsets of support $n\cdot 2^{O(d^3)}$ \cite{lubotzky2005explicit,kaufman2018construction}, bringing the support much closer to traditional `weakly-pseudorandom' subsets (satisfying e.g.\ high probability tail bounds, but not hypercontractivity) such as expander walks with dependence $n \cdot 2^{O(d)}$. Since support size typically corresponds to the `cost' or `blowup' of using $X$ in application (see e.g.\ their use in coding theory \cite{dinur2019list,alev2020list,jeronimo2021near}, derandomization and average-case hardness amplification \cite{impagliazzo1997p,impagliazzo2009approximate,impagliazzo2008uniform}, and PCPs \cite{impagliazzo2009new,dinur2011derandomized,bafna2024quasi}), constructing pseudorandom subsets of minimal support is a critical goal.

The special cases of \Cref{thm:Bonami-intro} and \Cref{cor:intro-operator-form} for product spaces have had numerous applications across boolean analysis \cite{lifshitz2019noise,keevash2021global}, extremal and additive combinatorics \cite{keevash2021global,keevash2023forbidden,keevash2023sharp}, and even algebra \cite{keller2023sharp,lifshitz2023bounds}. Following \cite{gur2021hypercontractivity,keevash2021global}, we give as a simple application a tight variant of the KKL-Theorem for HDX: low total influence functions are local.
\begin{corollary}[Low Influence Functions are Local (Informal \Cref{cor:KKL})]\label{cor:intro-KKL}
    Let $X$ be a $d$-partite $\gamma$-product with $\gamma\leq 2^{-\Omega(d)}$ and $f:X \to \mathbb{F}_2$ a constant variance function satisfying $\mathbf{I}[f] \leq K$. Then $\exists S \subset [d]$ with $|S| \leq O(K)$ and $x_S \in X[S]$ such that
    \[
    \mathbb{E}[f|_{x_S}] \geq 2^{-O(K)}.
    \]
\end{corollary}

\Cref{cor:intro-KKL} is only meaningful when $f$ is sparse (namely $\mathbb{E}[f] \leq 2^{-\Omega(K)}$). \cite{bafna2021hypercontractivity} raised the question of whether an analogous result could be proved for any density function on HDX. The classical such result on product spaces is Bourgain's booster theorem, which states that any low influence function has many restrictions on which it \textit{deviates} substantially from its expectation. We extend this result to HDX.

\begin{theorem}[A Booster Theorem for HDX (Informal \Cref{thm:booster})]\label{thm:intro-booster}
        Let $X$ be a $d$-partite $\gamma$-product with $\gamma \leq 2^{-\Omega(d)}$, and $f: X \to \{\pm 1\}$ a constant variance
        function satisfying $\mathbf{I}[f] \leq K$. Then
    \[
    \Pr_{x \sim X}\left[\exists T \subset [d]:~|T| \leq O(K),~|\mathbb{E}[f|_{x_T}] - \mathbb{E}[f]|>2^{-O(K^2)} \right] \geq 2^{-O(K^2)}
    \]
\end{theorem}
Classically, Bourgain's booster theorem lead to the famous theory of sharp thresholds for graph properties \cite{friedgut1999sharp}. It is an interesting open problem whether \Cref{thm:Bonami-intro} or \Cref{thm:intro-booster} could be used in this context.
\subsection{Further Related Work}\label{sec:intro-related-work}

\paragraph{Analysis Beyond the Cube:} Our work fits into a long line of boolean analysis and hypercontractive inequalities beyond the cube (see e.g.\ \cite{bourgain1992influence,talagrand1994russo,friedgut1996every,friedgut1998boolean,hatami2012structure,khot2018small,subhash2018pseudorandom,dikstein2018boolean,filmus2020hypercontractivity,bafna2020high,gur2021hypercontractivity,bafna2021hypercontractivity,gaitonde2022eigenstripping,keller2023sharp,ellis2024product} among many others), most closely aligned with work on product spaces and high dimensional expanders \cite{keevash2021global,zhao2021generalizations,keevash2023forbidden,bafna2021hypercontractivity,gur2021hypercontractivity,keller2023sharp}. In particular, we build on tools of \cite{gur2021hypercontractivity} who introduced the approximate Efron-Stein Decomposition for $\gamma$-products and proved a weaker $(2{\to}4)$-hypercontractive inequality roughly of the form
\begin{equation}\label{eq:old-GLL}
\norm{f^{\leq i}}_4^4 \leq i^{O(i)}\norm{f}^2_2\max_{|S| \leq i, x_S}\{\norm{f|_{x_S}}_2^2\} + 2^{O(d)}\gamma\norm{f}_2^2\norm{f}_\infty^2.
\end{equation}
As discussed in the previous sections, the main restriction of \Cref{eq:old-GLL} (beyond being in the basic $(2{\to}4)$-setting) is that the $i^{O(i)}$ dependence on degree does not scale appropriately with the eigenvalues of the noise operator. Achieving the correct scaling factor is crucial both in traditional applications of hypercontractivity and modern ones beyond the cube (e.g. similar improvements for the Grassmann were given in \cite{ellis2023analogue,evra2024polynomial} and used in  Minzer and Zheng's \cite{minzer2024near} recent breakthrough progress in low soundness PCPs), and was stated as the main open problem by one of the authors in \cite{liu2023hypercontractivity}.

It's also worth noting \Cref{eq:old-GLL} comes with an error term scaling in $\norm{f}_\infty$. While comparatively minor, this means the bound is only really meaningful when $\max_{|S| \leq i, x_S}\{\norm{f|_{x_S}}_2^2\} \geq 2^{O(d)}\gamma\norm{f}_\infty^2$. In other words, for any \textit{fixed} expansion parameter (or complex), the bound becomes trivial for sufficiently global functions. In contrast, \Cref{thm:Bonami-intro} holds uniformly for \textit{all} global functions.


\paragraph{Support Size and Application:}
Our work also draws on tools of \cite{bafna2021hypercontractivity} for two-sided local-spectral HDX. \cite{bafna2021hypercontractivity} prove a variant of \Cref{eq:old-GLL} with the correct dependence on degree $i$, but with similar loss in $\norm{f}_\infty$. As alluded to previously, the main limitation of their work lay in requiring a property known as $2^{-\Omega(d)}$-two-sided local-spectral expansion, all known constructions of which have support size $n\cdot 2^{2^{O(d)}}$. Our inequality holds for partite $2^{-\Omega(d)}$-one-sided HDX, which only require $n \cdot 2^{O(d^3)}$ support \cite{lubotzky2005explicit,kaufman2018construction}, an exponential improvement. We emphasize again that in application, this trade-off between dimension and support is by far the most critical parameter. For instance, in application to codes, higher dimension corresponds to better distance and decoding, while larger support size leads to worse rate \cite{dinur2019list,alev2020list}. Achieving optimal trade-offs in such parameters is a major open problem. The same is true for PCPs and average-case hardness amplification, where higher dimension leads to better soundness, while larger support leads to bigger size blowup \cite{impagliazzo1997p,impagliazzo2008uniform,impagliazzo2009new,dinur2011derandomized,dinur2017exponentially,minzer2024near,bafna2024quasi}. 

In this latter vein, it is worth discussing in slightly more detail the connection of our work to Bafna, Minzer, and Vyas' \cite{bafna2024quasi} recent use of HDX for quasilinear PCPs in the $1\%$ soundness regime. Their work gives the first improvement over Moshkovitz-Raz \cite{moshkovitz2008two} in nearly 20 years, but unfortunately has blowup scaling very poorly with soundness, precluding any improvement in the \textit{subconstant} soundness regime (needed, e.g., for applications to hardness of $\omega(1)$-approximating set cover). This is in part due to the fact that known agreement tests on HDX \cite{dikstein2024low,bafna2024characterizing}, a core component of such PCPs, are doubly exponentially far from having optimal soundness. Variants of hypercontractivity (namely global hypercontractivity of the Grassmann \cite{dinur2018towards,minzer2024near} and reverse hypercontractivity on $[n]^d$ \cite{dinur2017exponentially}) have played a core role in past works giving optimally sound testers and resulting optimal PCPs. It is not known how to similarly apply hypercontractivity on HDX to achieve such bounds, but such an argument would give a major step toward better subconstant soundness PCPs.

It is also worth highlighting that due to its use of symmetrization, the proof of \Cref{thm:Bonami-intro} is substantially simpler than prior proofs of hypercontractivity for HDX \cite{gur2021hypercontractivity,bafna2021hypercontractivity}. The proof of hypercontractivity on other extended spaces such as the Grassmann \cite{subhash2018pseudorandom,ellis2023analogue,evra2024polynomial} or Lie groups \cite{ellis2024product}, despite recent simplifications, still remains prohibitively complicated to extend to related objects such as tensors. Our result is the first to extend the simple symmetrization technique beyond products---could such a method also simplify and extend modern algebraic hypercontractive inequalities?

\paragraph{Coordinate-Wise Analysis on (near)-Product Spaces:} Coordinate-wise methods are among the most classical techniques for proving high dimensional analytic inequalities (see \cite{o2014analysis}). Some of these techniques, such as tensorization of variance and entropy, have already found great success in the study of approximate products and high dimensional expanders. This was first made explicit by Chen, Liu, and Vigoda \cite{chen2021optimal}, though similar ideas were explored earlier by Kaufman and Mass \cite{kaufman2020local}. Chen and Eldan \cite{chen2022localization} observed that many methods in the approximate sampling literature, such as spectral and entropic independence (variants of high dimensional expansion), can also be viewed from the of standpoint coordinate-wise localization schemes. 

Our coordinate-wise method is somewhat different than the above approaches. First, all coordinate-wise methods in the literature beyond the study of variance strongly relied on \textit{density} and \textit{balance} of the underlying complex. Indeed many of these works (see e.g.\ \cite{chen2021optimal,jain2023optimal} among others) actually prove hypercontractive inequalities for all functions, which cannot hold in the sparse setting. Second, unlike typical tensorization methods, we localize more explicitly at the level of the \textit{operator}, splitting $T_\rho$ itself into coordinate operators still acting over the full space $X$, then localizing one-by-one in a replacement style method. We do note that the former part of this approach was also considered in a somewhat different context in concurrent and independent work
of Alev and Parzanchevski \cite{alev2023sequential}, who analyzed a coordinate-wise decomposition of the down-up walk.
\subsection{Technical Overview}
Our proof of the symmetrization theorem relies on two elementary new tools in the theory of high dimensional expansion. First, we'll introduce the notion of `$q$-norm HDX' which bounds the local divergence of $X$ from a product in $q$-norm. Second, we introduce our coordinate-wise treatment of the noise operator and sketch how the technique is used to prove the symmetrization theorem. Finally, we sketch how symmetrization is used to prove optimal global hypercontractivity. We omit any technical overview of the booster theorem which is more technically involved but uses morally similar methods.

\paragraph{$(q,\gamma)$-Products:} A core problem identified in \cite{gur2021hypercontractivity} is that standard HDX and $\gamma$-products only bound the \textit{spectral} behavior of the complex, which a priori seems insufficient to control higher norm behavior. \cite{gur2021hypercontractivity} handle this issue in part by losing factors in infinity norm, e.g.\ by relating $q$ and $2$ norms via the elementary bound $\norm{f}_q^q \leq \norm{f}_2^2\norm{f}_\infty^{q-2}$, but this is too lossy to prove the symmetrization theorem.

We handle this and analogous issues by introducing $q$-norm HDX. Given a partite complex $X$, let $A_{i,j}$ denote the (normalized) bipartite adjacency matrix of $(X_i,X_j)$, and $\Pi_{i,j}$ be $A_{i,j}$'s stationary operator.\footnote{In this case, this is simply the matrix where every row is the marginal distribution over $X_i$. In the feasible conditioning, this is replaced with the conditioned marginal over $X_i$ (similarly for $A_{i,j}$).}
\begin{definition}[$(q,\gamma)$-Products]
    A $d$-partite complex $X$ is a $(q,\gamma)$-product if for every distinct $i,j \in [d]$:
    \begin{enumerate}
        \item The marginal $(X_i,X_j)$ is `$q$-norm-expander'
        \[
        \norm{A_{i,j} - \Pi_{i,j}}_q \leq \gamma
        \]
        \item This holds under all feasible conditionings $X_S=z_S$.
    \end{enumerate}
\end{definition}
Note that setting $q=2$ exactly recovers the notion of a $\gamma$-product. While $q$-norms beyond the spectral setting are typically much harder to analyze, an elementary application of Riesz-Thorin interpolation implies any sufficiently strong $\gamma$-product is also a $(q,\gamma_q)$-product.
\begin{lemma}[Informal \Cref{lem:2-to-p}] Any $\gamma$-product is a $(q,\gamma_q)$-product for $\gamma_q \leq \gamma^{2/q}2^{1-2/q}$.
\end{lemma}
This simple observation allows us to work directly with the $q$-norm, avoiding lossy factors in $\norm{f}_\infty$.

\paragraph{Coordinate-Wise Analysis on HDX:} The second critical tool in our analysis is a simplified method of coordinate-wise analysis on HDX inspired by Bourgain's proof of the symmetrization theorem. We focus in particular on the noise operator, though the same method can be applied to other random walks on HDX.

On a product space, the noise operator $T_\rho$ can naturally be expressed as the product of coordinate-wise operators $T_{\rho}^i$. When $\rho \in [0,1]$, this simply corresponds to re-sampling the $i$th coordinate with probability $1-\rho$ and keeping all other coordinates fixed. For general $r \in \R^d$, the corresponding coordinate noise operator is often defined as
\[
T_{r_i}^if \coloneqq \sum\limits_{S \ni i}r_if^{=S} + \sum\limits_{S \not\ni i} f^{=S}.
\]
Unfortunately, while this generalization makes manipulation of coordinate-wise noise operators simpler, it does not interact well with $\gamma$-products where the Efron-Stein decomposition is not closed (i.e.\ the Efron-Stein decomposition of $f^{=S}$ itself is not exactly $f^{=S}$). Instead, it turns out the right choice is to generalize $T_r$ in terms of the \textit{projection operators}
\[
E_Sf(y) = \mathbb{E}_x[f(x)~|~x_S=y_S],
\]
which re-sample all coordinates outside of $S$. The extended noise operator may then be defined as a `binomial' combination of projection operators:
\[
T_r = \sum\limits_{S \subseteq [d]} r_S \prod_{i \notin S}(1-r_i) E_S,
\]
where $r_S=\prod\limits_{i \in S}r_i$. The coordinate-wise noise operators are then naturally defined as 
\[
T_{r_i}^i \coloneqq r_iI + (1-r_i)E_{[d]\setminus i}.
\]
On a product space, it is easy to check that for either definition $T_r = T_{r_1}^1\ldots T_{r_d}^d$. Our first key lemma shows this continues to hold approximately on $(q,\gamma)$-products.
\begin{lemma}[Decorrelation (Informal \Cref{lemma:decorrelate})]\label{lemma:intro-decorrelate}
Let $X$ be a $d$-partite $(q,\gamma)$-product and $r\in \R^d$. Then
\[
\norm{T_r f - T^1_{r_1}\ldots T^d_{r_d} f}_q \leq O_\gamma(\norm{f}_{q})
\]
\end{lemma}
Second, we argue that for the operator formalization, the coordinate-wise operators \textit{localize} correctly to the $1$-D marginals of $X$.
\begin{lemma}[Localization (Informal \Cref{lem:localize})]\label{lem:intro-localize}
Let $X$ be a $d$-partite complex and $r \in \R^d$. Then for any $f:X \to \R$ and $i \in [d]$:
\[
T^i_{r} f(x) = T_{r_i} f|_{x_{-i}}(x_i),
\]
where $f|_{x_{-i}}: X[i] \to \R$ is the localization of $f$ defined by $f|_{x_{-i}}(x_i)=f(x)$.\footnote{Formally, this function lives on the \textit{link} of $x_{-i}$, see \Cref{sec:prelims}.}
\end{lemma}
Decorrelation and localization allows us to reduce analysis of the `high dimensional' quantity $T_\rho f(x)$ to the \textit{1-dimensional quantities} $T_\rho f|_{x_{-i}}(x_i)$. One can then apply standard $1$-dimensional inequalities, and `lift' the result back to the original complex $X$ by reversing the lemmas. This is a powerful (and standard) trick on the hypercube and product spaces in boolean analysis---many classical results are proved via this type of reduction to the $1$-D case (often properties proved this way, like hypercontractivity, are said to `tensorize').

\paragraph{The Symmetrization Theorem:}
We now sketch the proof of the symmetrization theorem based on these three components: a (standard) $1$-D version of the inequality, decorrelation into coordinate-wise operators, and localization to the $1$-D case. We will show only the upper bound; the lower bound follows via similar reasoning. The proof follows the outline of the original proof on products (as presented in \cite{o2014analysis}).

Slightly re-phrasing the theorem statement, our goal is to show
\[
\norm{T_{\frac{1}{2}}f}_q \leq (1+o_\gamma(1))\norm{T_r f}_q.
\]
By \Cref{lemma:intro-decorrelate}, we can break these terms into their coordinate-wise components and instead argue
\[
\norm{T^1_{\frac{1}{2}}\ldots T^d_{\frac{1}{2}}f}_q \leq (1+o_\gamma(1))\norm{T^1_{r_1}\ldots T^d_{r_d} f}_q.
\]
Written in this form, there is really only one natural strategy: the \textit{replacement method}. In particular, we'll argue that each of the coordinate-wise operators $T_{1/2}^i$ can be sequentially replaced by $T_{r_i}^i$, incurring only $O(\gamma)\norm{f}_q$ error in each step. Toward this end, define the `$j$-partially symmetrized' operator as
\[
T^{(j)} \coloneqq T^1_{r_1},\ldots T^j_{r_j}T^{j+1}_{\frac{1}{2}}\ldots T^d_{\frac{1}{2}}.
\]
Since $T^{(0)}=T^1_{1/2}\ldots T^d_{1/2} \approx T_\rho$, and $T^{(d)}=T^r_{r_1}\ldots T^d_{r_d} \approx T_r$ by telescoping it is enough to show for all $j$
\[
\norm{T^{(j)}f}_q \leq (1+o_\gamma(1))\norm{T^{(j+1)}f}_q.
\]
Un-wrapping the lefthand side, by \Cref{lemma:intro-decorrelate} we can permute the order of the $T_j$'s to move the $(j+1)$st operator to the front, and localize to the $(j+1)$st coordinate via \Cref{lem:intro-localize}
\begin{align*}
    \norm{T^{(j)}f}_q &= \norm{T^1_{r_1},\ldots T^j_{r_j}T^{j+1}_{\frac{1}{2}}\ldots T^d_{\frac{1}{2}}f}_q\\
    &\approx \norm{T^{j+1}_{\frac{1}{2}} (T^1_{r_1},\ldots T^j_{r_j}T^{j+2}_{\frac{1}{2}}\ldots T^d_{\frac{1}{2}}f)}_q\\
    &= \norm{\norm{T_{\frac{1}{2}} (T^1_{r_1},\ldots T^j_{r_j}T^{j+2}_{\frac{1}{2}}\ldots T^d_{\frac{1}{2}}f)|_{x_{-(j+1)}}}_{q,x_{j+1}}}_{q,x_{-(j+1)}}.
\end{align*}
The inner norm is now over a $1$-dimensional domain, so we can apply the standard $1$-D symmetrization theorem to replace $T_{1/2}g$ with $T_{r_{j+1}}g$, then `un'-localize and `un'-permute by the reverse directions of \Cref{lem:intro-localize} and \Cref{lemma:intro-decorrelate}:
\begin{align*}
    \norm{\norm{T_{\frac{1}{2}} (T^1_{r_1},\ldots T^j_{r_j}T^{j+2}_{\frac{1}{2}}\ldots T^d_{\frac{1}{2}}f)|_{x_{-(j+1)}}}_{q,x_{j+1}}}_{q,x_{-(j+1)}} &\leq \norm{\norm{T_{r_{j+1}} (T^1_{r_1},\ldots T^j_{r_j}T^{j+2}_{\frac{1}{2}}\ldots T^d_{\frac{1}{2}}f)|_{x_{-(j+1)}}}_{q,x_{j+1}}}_{q,x_{-(j+1)}}\\
    &= \norm{T^{j+1}_{r_{j+1}}T^1_{r_1},\ldots T^j_{r_{j}}T^{j+2}_{\frac{1}{2}}\ldots T^d_{\frac{1}{2}}f}_{q}\\
    &\approx \norm{T^1_{r_1},\ldots T^j_{r_{j+1}}T^{j+2}_{\frac{1}{2}}\ldots T^d_{\frac{1}{2}}f}_{q}\\
    &=\norm{T^{(j+1)}f}_q,
\end{align*}
as desired.
\paragraph{Global Hypercontractivity:} Our proof of global hypercontractivity builds on the elegant proof of O'Donnell and Zhao \cite{zhao2021generalizations} for $(2{\to}4)$-hypercontractivity of products. Here we largely focus on how to extend the $(2{\to}4)$-case to HDX and achieve the correct $2^{O(i)}$-scaling and remark when applicable how the result can be extended to near-sharp bounds for the $(2{\to}q)$-case.

The key to extending \cite{zhao2021generalizations}'s proof to HDX lies in combining symmetrization with $q$-norm variants of several techniques from \cite{gur2021hypercontractivity} and \cite{bafna2021hypercontractivity}. We start with the basic fact that the Efron-Stein decomposition is close to an eigenbasis of the noise operator in $4$-norm:
\begin{lemma}[Approximate Eigenbasis]\label{lem:apx-eigen-intro}
    Let $X$ be a $\gamma$-product. For any $f$ and $r \in \R^d$
    \[
    \norm{T_r f^{=S} - r_Sf^{=S}}_4 \leq O_{d,r}(\gamma)\norm{f}_4.
    \]
\end{lemma}
In fact, since we want errors scaling in $\ell_2$-norm rather than $\ell_4$, we'll also use the following stronger bound for global functions:
\[
    \norm{T_r f^{=S} - r_Sf^{=S}}_4 \leq O_{d,r}(\gamma)\norm{f}^{1/2}_2\max_{T \subseteq S, x_T}\{\norm{f|_{x_T}}_2^{1/2}\}.
\]

We now follow the core idea of \cite{zhao2021generalizations}: simply bound the 4-norm by symmetrizing and applying standard hypercontractivity. In particular by \Cref{thm:symmetrization-intro} and \Cref{lem:apx-eigen-intro}:
\begin{align*}
    \norm{f^{\leq i}}^4_4 &\lesssim \norm{T_rT_2f^{\leq i}}_{4}^4 & \text{(Symmetrization)}\\
    &\lesssim \norm{\left(\sum\limits_{|S| \leq i}r_S2^{|S|}f^{=S}\right)}^4_{4} & \text{($4$-Approximate Eigenbasis)}\\
    &=\mathbb{E}_x\left[\mathbb{E}_r\left[\left(\sum\limits_{|S| \leq i}r_S2^{|S|}f^{=S}(x)\right)^4\right]\right].
\end{align*}
Notice that the inner expectation, now over the boolean hypercube, is a degree-$i$ function with Fourier coefficients $2^{|S|}f^{=S}(x)$. Thus we may apply the standard Bonami Lemma and Parseval's Theorem to move from the $4$-norm to the $2$-norm:
\begin{align*}
        &\leq 2^{O(i)}\mathbb{E}_x\left[\mathbb{E}_r\left[\left(\sum\limits_{|S| \leq i}r_Sf^{=S}(x)\right)^2\right]^2\right] & \text{(Bonami Lemma)}\\
        & = \mathbb{E}_x\left[\left(\sum\limits_{|S| \leq i} f^{=S}(x)^2\right)\left(\sum\limits_{|T| \leq i} f^{=T}(x)^2\right)\right]. & \text{(Parseval)}
\end{align*}
If the two inner terms were independent, we'd be done, as each term individually is roughly the 2-norm $\norm{f^{\leq i}}_2^2$ by Parseval. Unfortunately the terms are correlated by shared variables within $x$. \cite{zhao2021generalizations}'s main insight is that if we condition on their intersection, the remaining variables can be handled independently at the cost of using the conditional norm of $f$ instead of the global norm for one of the factors.

More formally, we re-index our sum over intersections $I=S \cap T$:
\begin{equation}\label{eq:intro-hyp-full}
\mathbb{E}_x\left[\left(\sum\limits_{|S| \leq i} f^{=S}(x)^2\right)\left(\sum\limits_{|T| \leq i} f^{=T}(x)^2\right)\right] \leq \sum\limits_{I \leq i}\mathbb{E}_x\left[\sum\limits_{S \supset I} f^{=S}(x)^2\sum\limits_{T: T \cap S = I} f^{=T}(x)^2\right].
\end{equation}
On a product space, one could factor out the intersecting variables and bound the above as
\begin{equation}\label{eq:overview-tech-2}
\sum\limits_{I \leq i}\mathbb{E}_{x_I}\left[\left(\sum\limits_{S \supset I} \mathbb{E}_{x_{S \setminus I}}[f^{=S}(x_S)^2]\right)\left(\sum\limits_{T \supset I} \mathbb{E}_{x_{T \setminus I}}[f^{=T}(x_{T})^2]\right)\right].
\end{equation}
where critically the inner expectations are de-correlated. Once the bound is in this form, one argues that the maximum of the righthand term can be bounded by at most $2^{O(i)}$ times the worst conditional $2$-norm of $f$ over $x_I$. The remaining term is at most $2^{O(i)}$ times the global $2$-norm of $f$ which completes the proof. 

In the $(2{\to}q)$-case, the argument is similar but there are $q/2$ `$2$-norm' terms to handle after applying symmetrization and the Bonami Lemma. We argue it is still possible to pull out the final such term at the cost of $O(q)^i$ times the worst conditional $2$-norm. Repeating iteratively $q/2-1$ times results in a dependency of $O(q)^{qi/2}$, which combined with an additional factor of $O(q)^{qi/2}$ coming from applying the traditional Bonami Lemma at the start of the process gives the desired near-sharp $O(q)^{qi}$ scaling.


Returning to the $(2{\to}4)$-setting on $\gamma$-products, we remark one cannot naively simplify to \Cref{eq:overview-tech-2} because the variables $x_{S \setminus I}$ and $x_{T \setminus I}$ are not independent. To break the correlation, we'll use a standard tool in the HDX literature called the (partite) \textit{swap walks}:
\[
A_{S,T}f(y) = \mathbb{E}[f(x_S)~|~x_T=y_T],
\]
which correspond to the (normalized) bipartite adacency matrix of $(X_S,X_T)$. It is by now a widely used fact that the swap walks on HDX are highly expanding (see e.g.\ \cite{dikstein2019agreement,alev2019approximating,gur2021hypercontractivity,alev2023sequential,dikstein2024chernoff}). 

Formally, expanding out \Cref{eq:intro-hyp-full} one encounters many terms of the form
\[
    \underset{x_{S\setminus I}}{\mathbb{E}}\left[f^{=S}(x_S)^2\underset{x_{T \setminus I} \sim X_{x_S}}{\mathbb{E}}\left[f^{=T}(x_T)^2 \right]\right].
\]
The inner expectation is exactly an application of the swap walk in the link of $x_I$, so we can de-correlate the terms using spectral expansion as:
    \begin{align*}
        \underset{x_{S\setminus I}}{\mathbb{E}}\left[f^{=S}(x_S)^2\underset{x_{T \setminus I} \sim X_{x_S}}{\mathbb{E}}\left[f^{=T}(x_T)^2 \right]\right] &=\left\langle (f^{=S}|_{x_I})^2, A^{x_I}_{T\setminus I, S\setminus I}(f^{=T}|_{x_I})^2\right\rangle\\
        &\lesssim \underset{x_{S \setminus I} \sim X_{x_I}}{\mathbb{E}}[(f^{=S}|_{x_I})^2]\underset{x_{T \setminus I} \sim X_{x_I}}{\mathbb{E}}[(f^{=T}|_{x_I})^2]
    \end{align*}
and continue the proof as in the product case. Most of the actual work in the proof (once one has symmetrization) comes down to careful analysis of the error terms we've swept under the rug. With enough effort, we show the above inequalities hold up to error
$O_d(\gamma\norm{f}_2^2\max_{|S| \leq i, x_S}\{\norm{f|_{x_S}}_2^2\})$. This is dominated by the main term for small enough $\gamma$, so we may proceed with no asymptotic loss.

\section{Preliminaries}

\subsection{Simplicial Complexes}\label{sec:prelims}
A \textit{pure simplicial complex} is a collection of disjoint sets
\[
X = X(0) \cup \ldots \cup X(d)
\]
where $X(d) \subseteq {[n] \choose d}$ is an arbitrary $d$-uniform hypergraph and $X(i) \subseteq {[n] \choose i}$ is given by downward closure, that is the family of $i$-sets that sit inside any element of $X(d)$. A \textit{weighted} simplicial complex is a pure simplicial complex equipped with a measure $\Pi_d$ over $X(d)$. This induces a natural measure $\Pi_i$ over $X(i)$ given by drawing a $d$-set $s$, then an $i$-set $t \subset s$ uniformly at random. For simplicity, we will typically drop the weight function from notation and simply write $X$ to mean a pure weighted simplicial complex.

We denote the space of functions over $i$-faces of $X$ as $C_i \coloneqq \{f: X(i) \to \R\}$. The measures $\Pi_i$ induce a natural inner product over this space defined by
\[
\langle f,g \rangle = \underset{t \sim \Pi_i}{\mathbb{E}}[f(t)g(t)]
\]
We will usually just write $t \sim X(i)$, where it is understood that $\Pi_i$ is the underlying distribution.
\paragraph{Links:} It will frequently be useful to look at the local structure of a given simplicial complex $X$. Given a face $t \in X$, the \textit{link} of $t$ is the simplicial complex 
\[
X_t \coloneqq \{s: t \cup s \in X \land s \cap t = \emptyset\}
\]
whose underlying distribution $\Pi^t$ is the natural induced weight function generated by sampling $s' \in X(d)$ conditioned on $s'$ containing $t$, that is $\Pi^t(s) = \Pr_{s' \sim \Pi_d}[s'=s \cup t~|~t \subset s']$. Equivalently, one may think of the link $X_t$ as being the distribution over $X$ conditioned on $t$ appearing in the face, marginalized to the remaining variables.

\paragraph{Partite Complexes:} A simplicial complex $X$ is called \textit{partite} if it is possible to partition its vertices $X(0)=\Omega_1 \cup \ldots \cup \Omega_d$ such that every face $s \in X(d)$ has exactly one vertex from each component. We can always think of a partite simplicial complex as a (possibly very sparse) distribution over $\bigotimes\limits_{i=1}^d \Omega_i$. With this in mind, we will typically denote elements in $X(d)=\bigotimes\limits_{i=1}^d \Omega_i$ as $d$-dimensional tuples $x=(x_1,\ldots,x_d)$. This allows us to align with more typical notation in the analysis literature over product spaces.

It will frequently be useful to work on the \textit{projection} of $X(d)$ to a certain subset of colors (coordinates). In particular, given $S \subset [d]$, a face $x_S \in X[S]$ is simply generated by drawing $x \in X(d)$ and projecting onto the $S$-coordinates of $x$. Note these projections can be thought of as the marginals of $X=(x_1,\ldots,x_d)$. It is often useful to write nested expectations over the marginals of $x$, e.g.\ of the form $\mathbb{E}_{x_S}[[\mathbb{E}_{x_T}]...$ for disjoint $S$ and $T$. We emphasize that this notation always means $x_T$ is drawn \textit{conditionally} on the value of $x_S$ in the outer expectation.

\subsection{High Order Random Walks}
Hypergraphs come equipped with a sequence of natural random walks generalizing the standard random walk on graphs. In the context of high dimensional expanders, such walks were first studied by Kaufman and Mass \cite{kaufman2016high}, and have since become an integral part of almost all work on high dimensional expanders.

In this work, we focus mostly on partite complexes. In this setting, there is a natural family of generalizations of the bipartite graph operator. Given subsets $S,T \subset [d]$, the random walk operator $A_{S,T}$ maps functions on $X[S]$ to functions on $X[T]$ by averaging. In particular, given a function $f:X[S] \to \R$ define
\[
A_{S,T}f(y_T) = \mathbb{E}[f(x_S)~|~x_T=y_T].
\]
Equivalently, this is the expected value of $f$ when sampled from the link of $y_T$. Note that in the graph case, this is simply the underlying random walk on a bipartite graph. We write $\Pi_{S,T}$ to denote the stationary distribution of the operator $A_{S,T}$. Note this is simply the distribution given by drawing $x \in X(d)$ and projecting onto $S$.

We will frequently make use of the special case where $S=[d]$, which averages $f$ outside the specified coordinate subset $T$. In particular, given a subset $T \subset [d]$ and function $f: X(k) \to \R$, we define
\[
E_Tf(y_T) \coloneqq A_{[d],T}f(y_T) = \mathbb{E}[f(x)~|~x_T=y_T].
\]
Similarly, $E_Tf(y_T)$ is simply expectation of $f$ over the link of $y_T$.
\paragraph{The Noise Operator:} The noise operator is one of the best studied random walks in the analysis of boolean functions. The standard operator is typically defined via the following probabilistic interpretation describing its transition matrix as a Markov chain: given a face $x \in X(d)$, fix each coordinate in $x$ with probability $\rho$, and re-sample all remaining coordinates. Formally, we can write the noise operator as a convex combination of the $\{E_S\}_{S \subseteq [d]}$ operators defined above.
\begin{definition}[Noise Operator]
Let $X$ be a $d$-partite complex, $f \in C_d$, and $\rho \in [0,1]$. The noise operator $T_\rho$ acts on $f$ by re-randomizing over each coordinate with probability $1-\rho$:
\[
T_\rho f = \sum\limits_{S \subset [d]} \rho^{|S|}(1-\rho)^{d-|S|}E_S f
\]
\end{definition}
Finally, we note that when working on a link $X_\tau$, we will write any corresponding walk operator as $M^\tau$ to denote the specification to the link. E.g.\ $A_{i,j}^\tau$ is the bipartite graph between color $i$ and color $j$ within the link of $\tau$, and $\Pi_{i,j}^\tau$ to denote its stationary distribution.
\paragraph{Total Influence:} Total influence is a critical notion in boolean analysis measuring the total `sensitivity' of a function to flipping individual coordinates. On product spaces, total influence is most naturally defined via the \textit{Laplacian operators}
\[
L_i = I - A_{[d] \setminus \{i\}}.
\]
The total influence of a function $f \in C_d$ is 
\[
\mathbf{I}[f] = \sum\limits_{i \in [d]}\langle f, L_i f \rangle.
\]
When $f$ is boolean, it is easy to check that the inner product $\langle f, L_i f \rangle$ is proportional to the expected probability over $x$ that re-sampling $i$th coordinate flips the value of $f$. The study of the structure of low influence functions is core to boolean analysis, and a major motivation behind this work.

\subsection{High Dimensional Expanders and $\gamma$-Products}

We focus in this work on an elegant reformulation of spectral high dimensional expansion of Dikstein and Dinur \cite{dikstein2019agreement}, and Gur, Lifshitz, and Liu \cite{gur2021hypercontractivity} called a \textit{$\gamma$-product}, which promises every bipartite operator $A_{i,j}$ (in $X$ and its links) is expanding.
\begin{definition}[$\gamma$-Products]
A $d$-partite complex $X$ is called a $\gamma$-product if for every link $X_\tau$ of co-dimension at least $2$ and colors $i,j \notin \tau$, the spectral norm of the walk $A^\tau_{ij}$ satisfies: 
\[
\lambda_2(A^\tau_{ij}) \leq \gamma.
\]
\end{definition}
A similar notion was considered in the non-partite case by Bafna, Hopkins, Kaufman, and Lovett \cite{bafna2021hypercontractivity}. $\gamma$-products are closely related to the more standard notion of \textit{local-spectral expansion} of Dinur and Kaufman \cite{dinur2017high} and Oppenheim \cite{oppenheim2018local}.
\begin{definition}[Local-Spectral Expander]
A weighted complex $X$ is a (one-sided) $\gamma$-local-spectral expander if the graph underlying every non-trivial link is a (one-sided) $\gamma$-spectral expander.
\end{definition}
Dikstein and Dinur \cite{dikstein2019agreement} prove that any $\gamma$-one-sided partite local-spectral expander is an $O_d(\gamma)$-product. On the other hand, seminal work of Oppenheim \cite{oppenheim2018local} implies any $\gamma$-product is roughly an $O(\gamma)$-one-sided local-spectral expander for small enough $\gamma$ (see \cite{alev2023sequential} for a more detailed conversion between the two).

Gur, Lifshitz, and Liu \cite{gur2021hypercontractivity} also observe that any two-sided local-spectral expander can be \textit{embedded} into a partite complex as a $\gamma$-product. This is done by embedding of a $d$-dimensional complex $X$ into $X(1)^d$ by adding every permutation of each top level face.
\begin{observation}
If $X$ is a $d$-dimensional two-sided $\gamma$-local-spectral expander, then the partite complex\footnote{Technically, one should first fix an ordering on $X(1)$ for this to be well-defined, though the resulting complex is independent of choice of ordering.}
\[
X^d \coloneqq \{ \pi(x) : \pi \in S_d, x \in X(d)\}
\]
with distribution
\[
\mu^d(x) = \Pi_d(x)/d!
\]
is a $\gamma$-product.
\end{observation}
\begin{proof}
It is enough to observe that the color walk $A_{ij}$ is exactly the non-lazy upper walk on the graph underlying the original complex $X$.
\end{proof}
As a result, most results proved for $\gamma$-products carry over to two-sided local-spectral expanders. In fact, it turns out the approximate Fourier decomposition proposed by \cite{bafna2021hypercontractivity}, when passed through the above embedding, exactly corresponds to the Efron-Stein components. As a result, all Fourier analytic results on $\gamma$-products indeed transfer to two-sided HDX as well.

\section{Fourier Analysis on $(q,\gamma)$-Products}\label{sec:analysis-on-g-prods}
In this section, we introduce a natural generalization of local-spectral expanders to higher moments we call $(q,\gamma)$-Products, and extend the useful Fourier analytic machinery of \cite{gur2021hypercontractivity,bafna2021hypercontractivity} to this setting.
\begin{definition}[$(q,\gamma)$-Products]
For any $q \in (1,\infty]$, $\gamma >0$, and $d \in \mathbb{N}$, a $d$-partite complex $X$ is a $(q,\gamma)$-Product if for every link $X_\tau$ of co-dimension at least $2$ and every $i \neq j$
\[
\norm{A_{i,j}^\tau - \Pi_{i,j}^\tau}_{q} \leq \gamma.
\]
\end{definition}
Since $\norm{A^\tau_{i,j} - \Pi_{i,j}}_{2}$ is exactly the second largest eigenvalue of $A^\tau_{i,j}$, setting $p=2$ recovers the notion of a $\gamma$-product. One can define the notion of $(q,\gamma)$-HDX on non-partite complexes analogously.
\begin{definition}[$(q,\gamma)$-HDX]
A simplicial complex $X$ is a $(q,\gamma)$-HDX if every link $X_\tau$ of co-dimension at least $2$ satisfies:
\[
\norm{A_\tau - \Pi_1^\tau}_{q} \leq \gamma
\]
\end{definition}
Similarly, setting $q=2$ recovers the notion of a two-sided $\gamma$-local-spectral expander. It is easy to see that embedding a $(q,\gamma)$-HDX into a partite complex by including every ordering of the faces results in a $(q,\gamma)$-product by the same argument as above. All results we cover in the partite case therefore translate to the former, and we focus only on the partite case in what follows.

To relate $(q,\gamma)$-products to standard HDX, we will rely on the following special case of the classical Riesz-Thorin Interpolation theorem
\begin{theorem}[Riesz-Thorin Interpolation \cite{riesz1927maxima,thorin1939extension}]\label{thm:RT}
    Let $X$ be a $d$-partite complex and $T: C_d \to C_d$ a linear operator. For any $0 < p_0 < p_1 \leq \infty$, $\theta \in (0,1)$, and $\frac{1}{p_\theta}=\frac{1-\theta}{p_0}+\frac{\theta}{p_1}$
    \[
    \norm{T}_{p_\theta} \leq \norm{T}_{p_0}^{1-\theta}\norm{T}_{p_1}^{\theta}.
    \]
\end{theorem}
Because the infinity norm of the operators $A-\Pi$ is universally bounded, \Cref{thm:RT} implies any $\gamma$-product is a $(q,\gamma_q)$-product for $\gamma_q$ an appropriate function of $\gamma$ and $q$:
\begin{lemma}\label{lem:2-to-p}
    For any $q>1$ and $\gamma >0$, any $\gamma$-product is a $(q,\gamma_q)$-product for
    \[
    \gamma_q=\begin{cases}
        \gamma^{2/q}2^{1-2/q} & q \geq 2\\
        \gamma^{\frac{2(q-1)}{q}}2^{1-{\frac{2(q-1)}{q}}} & $1 < q < 2$
    \end{cases}
    \]
\end{lemma}
\begin{proof}
    It is enough to prove the $q \geq 2$ case. In particular, since the adjoint of $A_{i,j}$ is $A_{j,i}$, the $1<q<2$ case follows from H\"{o}lder duality, that for any $M$:
        \[
    \norm{M}_q = \norm{M^*}_{q'}
    \]
    where $q'=\frac{q}{q-1}$ is $q$'s H\"{o}lder conjugate. Assuming now that $q \geq 2$, for all $A_{i,j}^\tau$ we have:
    \[
    \norm{A_{i,j}^\tau - \Pi_{i,j}^\tau}_{2} \leq \gamma \quad \text{and} \quad \norm{A_{i,j}^\tau - \Pi_{i,j}^\tau}_{\infty} \leq 2,
    \]
    where the former is from definition and the latter is due to stochasticity of $A_{i,j}^\tau$ and $\Pi_{i,j}^\tau$. Riesz-Thorin Interpolation then gives
    \[
    \norm{A_{i,j}^\tau - \Pi^\tau}_{q} \leq \norm{A_{i,j}^\tau - \Pi^\tau}_{2}^{2/q}\norm{A_{i,j}^\tau - \Pi^\tau}^{1-2/q}_{\infty} \leq \gamma^{2/q}2^{1-2/q}
    \]
    as desired.
\end{proof}
\Cref{lem:2-to-p} allows us to easily control the higher norm behavior of classical constructions of high dimensional expanders such as the Ramanujan complexes \cite{lubotzky2005explicit} and coset complexes \cite{kaufman2018construction} over sufficiently large fields.

\subsection{$q$-Norm Expansion of Swap Walks}
A critical property of classical high dimensional expanders, and the core of \cite{gur2021hypercontractivity}'s Fourier analysis on HDX, is that the operators $A_{S,T}$ are strongly expanding when $S \cap T = \emptyset$ (these specific operators are often called the `swap-walks') \cite{alev2019approximating,dikstein2019agreement,gur2021hypercontractivity,alev2023sequential}. Prior arguments for expansion of the swap walks take specific advantage of being in the spectral ($q=2$) setting. We give an elementary argument showing swap walks expand in $q$-norm on $(q,\gamma)$-products. 
\begin{lemma}\label{lem:swap}
    Let $X$ be a $(q,\gamma)$-product and $q'=\frac{q}{q-1}$. For all coordinate subsets $S \cap T = \emptyset$:
        \[
        \norm{A_{S,T} - \Pi_{S,T}}_q \leq |S||T|\gamma \quad \text{and} \quad \norm{A_{S,T} - \Pi_{S,T}}_{q'} \leq |S||T|\gamma
        \]
\end{lemma}
\begin{proof}
    We induct on $|S|+|T|$. For the base case $|S|=|T|=1$, observe we have $\norm{A_{S,T}-\Pi_{S,T}}_q \leq \gamma$ by definition. To bound the $q'$-norm, recall H\"{o}lder duality promises that for any operator $M$
    \[
    \norm{M}_q = \norm{M^*}_{q'},
    \]
    where $M^*$ is the adjoint. Thus we also have
    \[
    \norm{A_{S,T} - \Pi_{S,T}}_{q'} = \norm{A_{T,S} - \Pi_{T,S}}_q \leq \gamma
    \]
    as desired. Note by the same argument, this also holds within every link of $X$.
    
    Assume now by induction that, for some fixed $|S|+|T|>2$, all $S',T'$ such that $|S'|+|T'| < |S|+|T|$ and $\tau \in X$ satisfy
    \[
    \norm{A^{\tau}_{S',T'} - \Pi^{\tau}_{S',T'}}_q \leq |S'||T'|\lambda.
    \]
    We first argue we may assume without loss of generality that $|T|>1$. This is again by H\"{o}lder duality since when $|S|>|T|$, we have
    \[
    \norm{A_{S,T}-\Pi_{S,T}}_q=\norm{A_{T,S}-\Pi_{T,S}}_{q'} \leq \gamma.
    \]
    Thus we are done if we show the result for all $q \in [1,\infty]$ and $|T|\geq |S|$ (in particular when $|T|>1$).

    Assume without loss of generality that $1 \in T$. Fix any $f: X[S] \to \R$. As is typically the case, the idea is to sample $t \in X[T]$ by first sampling $t' \in X[T \setminus \{1\}]$, then $v \in X_t[1]$:
    \begin{align*}
        \norm{A_{S,T}f - \Pi_{S,T}f}_q = \Bigg|\Bigg|\norm{(A_{S,T}f)|_{t'} - (\Pi_{S,T}f)|_{t'}}_{q,v \in X_t[1]}\Bigg|\Bigg|_{q,t' \in X[T \setminus \{1\}]}
    \end{align*}
    where for any $g: X[T] \to \R$, $g|_{t'}(v)=g(t' \cup v)$ denotes the localization of $g$ to the link of $t'$. The trick is now to observe that as a function of $X_{t'}[1]$, $A_{S,T}f|_{t'}$ is exactly $A_{S,1}^{t'}f^{t'}$, where $f^{t'}: X_{t'}[S] \to \R$ is the restriction $f^s(\tau)=f(\tau)$. Then by adding and subtracting the corresponding local stationary operator:
    \begin{align*}
         \norm{A_{S,T}f - \Pi_{S,T}f}_q &= \Bigg|\Bigg|\norm{A^{t'}_{S,1}f^{t'} - \Pi^{t'}_{S,1}f^{t'} + \Pi^{t'}_{S,1}f^{t'} - (\Pi_{S,T}f)|_{t'}}_{q,v \in X_{t'}(1)}\Bigg|\Bigg|_{q,t' \in X[T \setminus \{1\}]}\\
         &\leq \Bigg|\Bigg|\norm{A^{t'}_{S,1}f^{t'} - \Pi^{t'}_{S,1}f^{t'}}_{q,v \in X_{t'}[1]}\Bigg|\Bigg|_{q,s \in X[T\setminus \{1\}]}\\
         &+ \Bigg|\Bigg|\norm{\Pi^{t'}_{S,1}f^{t'} - (\Pi_{S,T}f)|_{t'}}_{q,v \in X_{t'}[1]}\Bigg|\Bigg|_{q,t' \in X[T \setminus \{1\}]}
    \end{align*}
    by the triangle inequality. The first term is now bounded by the inductive hypothesis applied in the link of $t'$:
    \begin{align*}
    \Bigg|\Bigg|\norm{A^{t'}_{S,1}f^{t'} - \Pi^{t'}_{S,1}f^{t'}}_{q,v \in X_{t'}[1]}\Bigg|\Bigg|_{q,t' \in X[T \setminus \{1\}]} &\leq i\gamma \Bigg|\Bigg|\norm{f^{t'}}_{q,X_{t'}[S]}\Bigg|\Bigg|_{q,t' \in X[T \setminus \{1\}]}\\
    &=\underset{t' \in X[T \setminus 1]}{\mathbb{E}}\left[\underset{s \sim X_{t'}[S]}{\mathbb{E}}[|f(s)|^q]\right]^{1/q}\\
    &= |S|\gamma \norm{f}_{q}
    \end{align*}
    where in the final step we have used the fact that drawing $t' \in X[T \setminus 1]$, then $s$ from $X_{t'}[S]$ is equidistributed with simply drawing $s \in X[S]$ directly.
    
    Toward the second term, observe that, as a function of $t'$:
    \[
    \Pi^{t'}_{S,1}f^{t'}(v) = A_{S,T'}f(t') \quad ~\text{and}~ \quad (\Pi_{S,T}f)|_{t'}(v)=\Pi_{S,T'}f(t')
    \]
    so we finally have
    \begin{align*}
    \Bigg|\Bigg|\norm{\Pi^{t'}_{S,1}f^{t'} - (\Pi_{S,T}f)|_{t'}}_{q,v \in X_{t'}[1]}\Bigg|\Bigg|_{q,t' \in X[T'\setminus \{1\}])} &= \norm{A_{S,T'}f - \Pi_{S,T'}f}_{q,t' \in X[T\setminus \{1\}]}\\ 
    &\leq (|T|-1)|S|\gamma \norm{f}_{q}
    \end{align*}
    by the inductive hypothesis. Altogether this gives
    \[
    \norm{A_{S,T}f - \Pi_{S,T}f}_q \leq \gamma |S| + \gamma (|T|-1)|S| = |S||T|\gamma
    \]
    as desired.
\end{proof}
\subsection{The Efron-Stein Decomposition}
Most of our analysis of $(q,\gamma)$-products is based on the Efron-Stein decomposition, introduced in the context of HDX independently in \cite{gur2021hypercontractivity} (for $\gamma$-products) and \cite{bafna2021hypercontractivity} (for two-sided HDX).
\begin{definition}[Efron-Stein Decomposition]
Let $X$ be a $d$-partite complex, $0 \leq k \leq d$, and $f \in C_k$. The Efron-Stein Decomposition of f is given by the collection of functions $\{f^{=S}\}_{S \subseteq [d]}$ where:
\[
f^{=S} \coloneqq \sum\limits_{T \subset S} (-1)^{|S \setminus T|}E_Tf
\]
We write $f^{\leq i}=\sum\limits_{|S| \leq i}f^{=S}$ to denote the degree at most $i$ components of $f$, and $f^{=i}$ degree exactly $i$.
\end{definition}
Note that $f^{=S}$ is only a function of the $S$-projection of $x$. With this in mind, we sometimes write $f^{=S}(x_S)$ instead of $f^{=S}(x)$. 

In the remainder of the section, we overview useful properties of the Efron-Stein decomposition that continue to hold approximately on $(q,\gamma)$-products. We start with by stating a core result of \cite{gur2021hypercontractivity}, who show the Efron-Stein decomposition on $\gamma$-products is an approximate Fourier basis.
\begin{theorem}[{\cite{gur2021hypercontractivity}}]\label{thm:apx-Efron-Stein}
Let $X$ be a $d$-partite $\gamma$-product and $f:X \to \R$. The Efron-Stein Decomposition satisfies the following properties:
\begin{enumerate}
    \item Decomposition: $f= \sum\limits_{S \subset [d]}f^{=S}$
    \item Approximate Orthogonality: $\langle f^{=S}, f^{=T} \rangle \leq 2^{O(|S|+|T|)}\gamma \norm{f}_2^2$
    \item Approximate Parseval: $\left | \langle f^{\leq i},f^{\leq i} \rangle - \sum\limits_{|S| \leq i} \langle f^{=S}, f^{=S} \rangle \right | \leq 2^{O(d)}\gamma\norm{f}_2^2$
\end{enumerate}
\end{theorem}


We will require variants of a few other standard properties of Efron-Stein extended to high norms. First, we observe the $q$-norm of $f^{=S}$ is not much larger than the $q$-norm of $f$ itself.
\begin{lemma}\label{lemma:Efron-contract}
Let $X$ be a $d$-partite complex. Then for any $f \in C_d$, $S \subseteq [d]$, and $q \geq 1$
\[
\norm{f^{=S}}_q \leq 2^{|S|}\norm{f}_q.
\]
The following $\ell_2$-bound also holds:
\[
\norm{f^{=S}}_q \leq 2^{O(|S|)}\norm{f}_2^{2/q}\max_{T \subseteq S, x_{T}}\left\{\norm{f|_{x_T}}_2^{1-\frac{2}{q}}\right\}
\]
\end{lemma}
\begin{proof}
$E_T$ is an averaging operator and therefore contracts $q$-norms. As such we have:
\begin{align*}
    \norm{f^{=S}}_q &= \norm{\sum\limits_{T \subset S}(-1)^{|S|-|T|}E_Tf}_q\\
    &\leq 2^{|S|}\norm{f}_q
\end{align*}
by the triangle inequality. For the second bound, we can instead use Cauchy-Schwarz to write:
\begin{align*}
    \norm{\sum\limits_{T \subset S}(-1)^{|S|-|T|}E_Tf}_q &= \mathbb{E}\left[\left(\sum\limits_{T \subset S}(-1)^{|S|-|T|}E_Tf\right)^q\right]^{1/q}\\
    &\leq 2^{O(|S|)}\mathbb{E}\left[\left(\sum\limits_{T \subset S}(E_Tf)^2\right)^{q/2}\right]^{1/q}\\
    &\leq 2^{O(|S|)}\max_{T \subseteq S, x_T}\left\{(E_Tf)^2\right\}^{\frac{1}{2}-\frac{1}{q}}\left(\sum\limits_{T \subset S}\mathbb{E}\left[(E_Tf)^{2}\right]\right)^{1/q}\\
    &\leq 2^{O(|S|)}\max_{T \subseteq S, x_T}\left\{\norm{f|_{x_T}}_{2}^{1-\frac{2}{q}}\right\}\norm{f}_2^{2/q}
\end{align*}
where we've used the fact that $\mathbb{E}[f|_{x_T'}]^2 \leq \norm{f|_{x_{T'}}}_2^2$ and $E_{T}$ contracts $p$-norms.
\end{proof}
Second, we will crucially rely on the fact that $E_Tf^{=S}$ is small so long as $T$ does not contain $S$.
\begin{lemma}\label{lemma:p-to-q-down}
Let $X$ be a $d$-partite $(q,\gamma)$-product and $f \in C_d$. Then for any $T,S \subset [d]$ such that $T \not\supset S$:
\[
\norm{E_Tf^{=S}}_q \leq |T|2^{O(|S|)}\gamma \norm{f}_q.
\]
If $X$ is a $\gamma$-product, the following $\ell_2$-bound also holds:
\[
\norm{E_Tf^{=S}}_q \leq |T|2^{O(|S|)}\gamma \norm{f}^{2/q}_2\max_{T' \subseteq S, x_{T'}}\left\{\norm{f|_{x_T'}}_2^{1-\frac{2}{q}}\right\}
\]
\end{lemma}
\begin{proof}
The proof is similar to the $q=2$ case in \cite{gur2021hypercontractivity}, but with somewhat more care required for the latter bound. The key is to prove the following claim
\begin{claim}\label{claim:intersect-norm}
In the $(q,\gamma)$-product case we have for any $f:X(d) \to \R$:
\[
\norm{E_{T}E_{T'}f - E_{T \cap T'}f}_q \leq \gamma|T||T'|\norm{f}_q.
\]
Similarly for the $\gamma$-product case
\[
\norm{E_{T}E_{T'}f - E_{T \cap T'}f}_q \leq \gamma|T||T'|\norm{f}_{2}^{2/q}\max_{T'' \subset T',x_{T''}}\left\{\norm{f|_{x_T''}}_2^{1-\frac{2}{q}}\right\}
\]
\end{claim}
\begin{proof}
First, we observe it is sufficient to consider the setting where $T \cap T' = \emptyset$. We prove this for the latter bound --- the former follows from an analogous argument similar to the $q=2$ case in \cite{gur2021hypercontractivity}. Assume the claim holds when $T \cap T' = \emptyset$, and let $I=T \cap T'$ for a generic choice of $T,T'$. We localize within $I$ and apply the claim:
\begin{align*}
    \norm{E_{T}E_{T'}f - E_{T \cap T'}f}_q &= \underset{x_I}{\mathbb{E}}\left[\underset{x_{\bar{I}}}{\mathbb{E}}\left[(E^{x_{I}}_{T \setminus I}E^{x_{I}}_{T' \setminus I} - E^{x_I}_{\emptyset})f|_{x_I})^q\right]\right]^{1/q}\\    
    &= \underset{x_I}{\mathbb{E}}\left[\norm{(E^{x_{I}}_{T \setminus I}E^{x_{I}}_{T' \setminus I} - E^{x_I}_{\emptyset})f}_{x_I,q}^q\right]^{1/q}\\ 
    &\leq \gamma |T||T'|\underset{x_I}{\mathbb{E}}\left[\norm{f|_{x_I}}_{2}^{2}\max_{T'' \subset T' \setminus I,x_{T''}}\left\{\norm{f|_{x_{T'' \cup I}}}_2^{q-2}\right\}\right]^{1/q}.
\end{align*}
We can now pull the inner maximum out by its maximum over $x_I$ as well to upper bound this as
\begin{align*}
    &\leq \gamma |T||T'|\max_{T'' \subset T',x_{T''}}\left\{\norm{f|_{x_{T''}}}_2^{1-\frac{2}{q}}\right\}\underset{x_I}{\mathbb{E}}\left[\norm{f|_{x_I}}_{2}^{2}\right]^{1/q}\\
    &=\gamma |T||T'|\max_{T'' \subset T',x_{T''}}\left\{\norm{f|_{x_{T''}}}_2^{1-\frac{2}{q}}\right\}\norm{f}_2^{2/q}
\end{align*}
as desired.

We therefore turn our attention to the case $T \cap T' = \emptyset$. Here we use the observation of \cite{gur2021hypercontractivity} that one may write $E_TE_{T'} = A_{T',T}E_{T'}$ and $\mathbb{E}[E_{T'} f]=E_{\emptyset}[f]$. Thus to prove the $q$-norm variant of the bound we simply write
\begin{align*}
    \norm{E_TE_{T'}f - E_{\emptyset}f}_q &= \norm{(E_{T',T}- \Pi_{T',T})E_{T'}f]}_q\\
    &\leq \gamma|T'||T| \norm{E_{T'}f}_q\\
    &\leq \gamma|T'||T| \norm{f}_q
\end{align*}
since averaging contracts $q$-norms. 

The $\ell_2$-bound follows similarly, but pulling out a factor of the conditional $2$-norm as in \Cref{lemma:Efron-contract}. Namely we can instead bound $\norm{E_{T'}f}_q$ as
\begin{align*}
    \norm{E_{T'}f}_q &= \mathbb{E}_x\left[(E_{T'}f)^q\right]^{1/q}\\
    &\leq \max_x\{E_{T'}f(x)^{1-\frac{2}{q}}\}\mathbb{E}_x\left[(E_{T'}f)^2\right]^{1/q}\\
    &\leq \max_{x_{T'}}\left\{\norm{f|_{x_{T'}}}_2^{1-\frac{2}{q}}\right\}\norm{f}_2^{2/q}
\end{align*}
\end{proof}
The result now follows simply by expanding $f^{=S}$:
\begin{align*}
    \norm{E_Tf^{=S}}_q &=\norm{E_T\sum\limits_{T' \subset S}(-1)^{|S|-|T'|}E_{T'}f}_q\\
    &= \norm{\sum\limits_{T' \subset S}(-1)^{|S|-|T'|}E_TE_{T'}f}_q\\
        &= \norm{\sum\limits_{T' \subset S}(-1)^{|S|-|T'|}E_{T \cap T'}f + \sum\limits_{T' \subset S}(-1)^{|S|-|T'|}(E_TE_{T'}-E_{T \cap T'})f}_q\\
    &\leq 2^{O(|S|)}\norm{(E_TE_{T'}-E_{T \cap T'})f}_q
\end{align*}
where we have used the standard fact that $\sum\limits_{T' \subset S}(-1)^{|S|-|T'|}E_{T \cap T'}$ is identically zero when $T \not\supset S$. Plugging in the claim now gives the desired results.
\end{proof}

A critical property of the standard Efron-Stein decomposition on product spaces is that it is an eigenbasis for standard random walks such as the noise operator. Using the above lemma, it is elementary to show this extends approximately to $(q,\gamma)$-products.
\begin{lemma}\label{lem:apx-eigen}
    Let $X$ be a $d$-partite $(q,\gamma)$-product and $M=\sum\limits_{T \subset [d]} \alpha_T E_T$ for some $\alpha_T \in \R$. Then for $f \in C_d$ and $S \subset [d]$:
    \[
    \norm{Mf^{=S} - \lambda_Sf^{=S}}_q \leq |\alpha|_12^{O(d)}\gamma\norm{f}_q
    \]
    where $\lambda_S= \sum\limits_{T \supseteq S}\alpha_T$ and $|\alpha|_1$ is the counting norm of $\alpha$. Moreover if $X$ is a $\gamma$-product we also have:
    \[
    \norm{Mf^{=S} - \lambda_Sf^{=S}}_q \leq |\alpha|_12^{O(d)}\gamma\norm{f}_2^{2/q}\max_{T \subseteq S, x_T}\left\{\norm{f|_{x_T}}_2^{1-\frac{2}{q}}\right\}
    \]
\end{lemma}
\begin{proof}
    Expanding $Mf^{=S}$ we have
    \begin{align*}
        Mf^{=S} &= \sum\limits_{T \subset [d]} \alpha_T E_Tf^{=S}\\
        &= \sum\limits_{T \supseteq S} \alpha_T f^{=S} + \sum\limits_{T \not\supseteq S} \alpha_T E_Tf^{=S}
    \end{align*}
    where the second equality holds since $f^{=S}$ depends only on coordinates in $S$, and $E_T$ fixes such coordinates when $T \supseteq S$. Applying \Cref{lemma:p-to-q-down} and the triangle inequality to the latter term completes the proof.
\end{proof}
A similar application also shows that the Efron-Stein basis is `approximately closed'.
\begin{lemma}\label{lemma:apx-closed}
    Let $X$ be a $d$-partite $(q,\gamma)$-product. Then for any $f \in C_d$ and $T,S \subseteq [d]$
    \begin{itemize}
        \item For $S \neq T$: 
        \[
        \norm{(f^{=S})^{=T}}_q \leq 2^{O(d)}\gamma\norm{f}_q
        \]
        \item For $S=T$:
        \[
        \norm{(f^{=S})^{=S} - f^{=S}}_q \leq 2^{O(d)}\gamma\norm{f}_q
        \]
    \end{itemize}
For $X$ a $\gamma$-product we may similarly replace $\norm{f}_q$ by $\norm{f}_2^{2/q}\underset{T \subseteq S, x_T}{\max}\left\{\norm{f|_{x_T}}_2^{1-\frac{2}{q}}\right\}$.
\end{lemma}
\begin{proof}
    Using the definition of Efron-Stein, we have
    \begin{align*}
        (f^{=S})^{=T} = \sum\limits_{T' \subseteq T} (-1)^{T \setminus T'}E_{T'}f^{=S}
    \end{align*}
    If $S \neq T$, then $T' \not\supset S$ for every term in the sum, and we are done by the triangle inequality and \Cref{lemma:p-to-q-down}. If $S=T$, the only `surviving' term is $E_Sf^{=S}=f^{=S}$, and we are similarly done.
\end{proof}
Finally, we will need the following behavior of the Efron-Stein decomposition under restriction.
\begin{lemma}\label{lem:Efron-Restrict}
Let $X$ be a $d$-partite complex and $f \in C_d$. Let $I,B \subset [d]$ be any two disjoint sets. Then
\[
f^{=I \cup B}(y_I,x_B,z) = \sum\limits_{J \subset I} (-1)^{|I|-|J|}(f|_{y_J})^{=B}(x_B,z).
\]
\end{lemma}
\begin{proof}
This is a standard fact on product spaces. We give a proof that generalizes to any partite complex. Expanding the righthand side:
\begin{align*}
    \sum\limits_{J \subset I}(-1)^{|I|-|J|}(f|_{y_J}^{=B})(x_B,z) &= \sum\limits_{J \subset I}(-1)^{|I|-|J|}\sum\limits_{A \subset B}(-1)^{|B|-|A|} E^{y_J}_Af|_{y_J}(x_B,z)\\
    &=\sum\limits_{J \subset I}(-1)^{|I|-|J|}\sum\limits_{A \subset B}(-1)^{|B|-|A|} \underset{z' \sim X_{y_J \cup x_A}}{\mathbb{E}}[f(y_J,x_A,z')]\\
    &=\sum\limits_{J \subset I}(-1)^{|I|-|J|}\sum\limits_{A \subset B}(-1)^{|B|-|A|} E_{J\cup A}f(y_I,x_B,z)\\
    &=\sum\limits_{T \subset I \cup B} (-1)^{|I\cup B| -|T|}E_Tf(y_I,x_B,z)\\
    &=f^{=I \cup B}(y_I,x_B,z)
\end{align*}
as desired.
\end{proof}
\subsection{Efron-Stein, Total Influence, and the Noise Operator}
Our work relies heavily on the following operator-wise generalization of the noise operator.
\begin{definition}[Vector-Valued Noise Operator]
Let $X$ be a $d$-partite complex, $f \in C_d$, and $r \in \R^d$. The noise operator $T_{r}$ is defined as:
\[
T_r f = \sum\limits_{S \subset [d]} \prod_{i\in S}r_i \prod_{i \notin S}(1-r_i)E_S f
\]
\end{definition}
While the individual Efron-Stein components are only approximate eigenvectors of the above operator, it turns out the full decomposition behaves exactly as in the product case.
\begin{claim}\label{claim:Efron-Noise}
Let $X$ be a $d$-partite complex, $f \in C_d$, and $r \in \R^n$. Then
\[
T_{r} f = \sum\limits_{S \subset [d]} r_S f^{=S},
\]
where $r_S = \prod\limits_{i \in S}r_i$
\end{claim}
\begin{proof}
    By inclusion-exclusion, as in the product setting we have
    \[
    E_S f = \sum_{T \subseteq S} f^{=T}.
    \]
    With this in hand, expanding $T_r f$ gives:
    \begin{align*}
        T_r f &= \sum\limits_{S \subset [d]} \prod_{i\in S}r_i \prod_{i \notin S}(1-r_i)E_S f\\
        &=\sum\limits_{S \subset [d]} \prod_{i\in S}r_i \prod_{i \notin S}(1-r_i)\sum\limits_{T \subseteq S} f^{=T}\\
        &= \sum\limits_{T \subset [d]} f^{=T}\left(\sum\limits_{S \supset T} \prod_{i\in S}r_i \prod_{i \notin S}(1-r_i)\right) & \text{(Re-indexing)}\\
        &=\sum\limits_{T \subset [d]} r_Tf^{=T}\left(\sum\limits_{S \supset T} \prod_{i\in S \setminus T}r_i \prod_{i \notin S}(1-r_i)\right)\\
        &=\sum\limits_{T \subset [d]} r_Tf^{=T}
    \end{align*}
    where in the final step we have observed that
    \[
            \sum\limits_{S \supset T} \prod_{i\in S \setminus T}r_i \prod_{i \notin S}(1-r_i)=\prod_{i \in [d] \setminus T}(1-r_i+r_i)=1.
    \]
\end{proof}
One can also view this as an equivalent method of defining the operator $T_r$. We will move freely between these two equivalent notions in what follows. We can similarly express the Laplacians in terms of Efron-Stein. The proofs are standard  (see e.g.\ \cite[Lemma 6.2]{gur2021hypercontractivity}) and similar to the above so we omit them.
\begin{lemma}
    Let $X$ be a $d$-partite complex and $f \in C_d$. Then:
    \begin{enumerate}
        \item $L_i f = \sum\limits_{S \ni i}f^{=S}$
        \item $\sum\limits_{i \in [d]} L_i f = \sum\limits_{S \subseteq [d]}|S|f^{=S}$
    \end{enumerate}
\end{lemma}
This implies the following useful relation between total influence and the Efron-Stein decomposition.
\begin{corollary}\label{claim:total-inf}
    Let $X$ be a $d$-partite complex and $f \in C_d$. Then
    \[
    \mathbf{I}[f] = \sum\limits_{i=0}^d i\langle f, f^{=i}\rangle
    \]
\end{corollary}


\section{Coordinate-Wise Analysis on HDX}\label{sec:coordinate-method}
The cornerstone of our analysis is a new elementary method of coordinate-wise analysis on high dimensional expanders inspired by the proof of Bourgain's symmetrization theorem. We break the argument into two main parts, a \textit{de-correlation} step where we break the noise operator into coordinate-wise components, and a \textit{localization} step where we show the coordinate-wise noise operators may be viewed as the standard noise operator on the projection/localization to the relevant coordinates. We will see in the next section that these lemmas facilitate classic analytic tools such as the replacement method.

We first define a coordinate-wise version of the noise operator that operates only over a specific subset.
\begin{definition}[Coordinate-Wise Noise Operator]
Let $X$ be a $d$-partite complex, $S \subseteq [d]$, and $r \in \R^S$. The coordinate-wise noise operator $T^S_r$ acts on $f$ by:
\[
T_r^S f \coloneqq \sum\limits_{T \subseteq S} r_{S \setminus T}\prod_{i \in T}(1-r_i) E_{[d] \setminus T}f
\]
\end{definition}
Note that we have inverted the sum in the sense that index $T$ corresponds to $E_{[d] \setminus T}$. When $r \in [0,1]^d$, the above should still be thought of as re-sampling each coordinate within $S$ with probability $1-r_i$.

Our first key lemma shows that $T_r$ can be approximately decomposed into its constituent coordinate operators (under any ordering) up to some error in $q$-norm. For simplicity of notation, given a permutation $\pi \in S_d$ and vector $r \in \R^d$ we write
\[
T^\pi_r \coloneqq T^{\pi(1)}_{{r_\pi(1)}}\ldots T^{\pi(d)}_{r_{\pi(d)}}.
\]
\begin{lemma}[Decorrelation]\label{lemma:decorrelate}
Let $X$ be a $d$-partite $(q,\gamma)$-product, $r\in \R^d$, and $\pi \in S_d$ any permutation, then:
\[
\norm{T_r f - T^\pi_r f}_q \leq c_{d,r}\gamma \norm{f}_{q}
\]
where $c_{d,r} = d^3\sum\limits_S \left |r_S \prod\limits_{i \notin S}(1-r_i)\right|$
\end{lemma}
\begin{proof}
The proof is essentially an immediate application of \Cref{claim:intersect-norm}. Assume $\pi=I$ without loss of generality (the analysis is invariant under ordering). Expanding out the coordinate-wise product we have
\begin{align*}
    T^{1}_{r_1}\ldots T^{d}_{r_d} &= \prod\limits_{i \in [d]}(r_{i} I + (1-r_{i})E_{[d]\setminus i})\\
    &=\sum\limits_{S} \left(r_S \prod_{i \notin S}(1-r_i)\right)\prod_{i \notin S}E_{[d]\setminus i}.
\end{align*}
It is therefore enough to bound the $q$-norm
\[
\norm{E_{S} - \prod_{i \notin S}E_{[d]\setminus i}}_q \leq d^3\gamma.
\]
This follows from iterated application of \Cref{claim:intersect-norm}. More formally, assume without loss of generality that $S=\{1,\ldots,j\}$. Then we can write the following telescopic sum:
\begin{align*}
    E_S - \prod_{i \notin S}E_{[d]\setminus i} &= \sum\limits_{\ell=1}^{d-j} \left(\left(\prod_{i=j+1}^{d-\ell}E_{[d] \setminus i}\right)E_{[d-\ell]} - \left(\prod_{i=j+1}^{d-\ell+1}E_{[d] \setminus i}\right) E_{[d-\ell+1]}\right)\\
    &= \sum\limits_{\ell=1}^{d-j} \left(\left(\prod_{i=j+1}^{d-\ell}E_{[d] \setminus i}\right)E_{[d-\ell]} - \left(\prod_{i=j+1}^{d-\ell}E_{[d] \setminus i}\right)E_{[d] \setminus 
 \{d-\ell+1\}}E_{[d-\ell+1]}\right)\\
 &=\sum\limits_{\ell=1}^{d-j} \left(\prod_{i=j+1}^{d-\ell}E_{[d] \setminus i}\right)\left(E_{[d-\ell]} - E_{[d] \setminus 
 \{d-\ell+1\}}E_{[d-\ell+1]}\right)
\end{align*}
By \Cref{claim:intersect-norm} and the triangle inequality, we finally have:
\begin{align*}
    \norm{E_S - \prod_{i \notin S}E_{[d]\setminus i}}_q &\leq \sum\limits_{\ell=1}^{d-j} \norm{\left(\prod_{i=j+1}^{d-\ell}E_{[d] \setminus i}\right)\left(E_{[d-\ell]} - E_{[d] \setminus 
 \{d-\ell+1\}}E_{[d-\ell+1]}\right)}_q\\
 &\leq \sum\limits_{\ell=1}^{d-j}\norm{E_{[d-\ell]} - E_{[d] \setminus 
 \{d-\ell+1\}}E_{[d-\ell+1]}}_q\\
 &\leq d^3\gamma
\end{align*}
where we have additionally taken advantage of the fact that averaging operators contract $q$-norm.
\end{proof}

The second critical lemma is our `localization' process, allowing us to view the coordinate-wise noise operators as local copies of the full noise operator on the links of $X$.
\begin{lemma}[Localization]\label{lem:localize}
Let $X$ be a $d$-partite simplicial complex, and $r \in \R^d$. Then for any $S \subseteq [d]$:
\[
T^S_{r} f(x) = T^{x_{\bar{S}}}_{r} f|_{x_{\bar{S}}}(x_S)
\]
\end{lemma}
\begin{proof}
    The proof is essentially immediate from the fact that the partite averaging operators `respect' restriction, namely for any $T \subseteq S$ we have by definition:
    \[
    E_{[d] \setminus T}f(x) = E^{x_{\bar{S}}}_{S \setminus T}f|_{x_{\bar{S}}}(x_S).
    \]
    Then expanding out $T_r^Sf$ gives:
    \begin{align*}
        T_r^S f(x) &= \sum\limits_{T \subseteq S} r_{T\setminus S}\prod_{i \in T}(1-r_i) E_{[d] \setminus T}f(x)\\
        &=\sum\limits_{T \subseteq S} r_{T\setminus S}\prod_{i \in T}(1-r_i)E^{x_{\bar{S}}}_{S \setminus T}f|_{x_{\bar{S}}}(x_S)\\
        &= T^{x_{\bar{S}}}_{r} f|_{x_{\bar{S}}}(x_S)
    \end{align*}
as desired.
\end{proof}


\section{The Symmetrization Theorem}
We are now ready to turn to our first major result, Bourgain's symmetrization theorem for HDX.

\begin{theorem}[Symmetrization on HDX]\label{thm:symmetrization}
Let $q > 1$ and $X$ be a $d$-partite $(q,\gamma)$-product for $\gamma \leq 2^{-\Omega(d)}$. Then for any $f \in C_d$:
\[
(1-2^{O(d)}\gamma)\norm{\widetilde{T_{c_q}f}}_q \leq \norm{f}_q \leq (1+2^{O(d)}\gamma) \norm{\widetilde{T_2f}}_q
\]
for some constant $0 \leq c_q \leq 1$ dependent only on $q$. 
\end{theorem}
Combined with \Cref{lem:2-to-p}, we immediately get symmetrization for standard $\gamma$-products and HDX.
\begin{corollary}\label{cor:sym-products}
Let $q \geq 1$ and $X$ be a $\gamma$-product with $\gamma \leq 2^{-\Omega(\max\{q,q'\}d)}$. For any $f \in C_d$:
\[
\left(1-2^{O(d)}\gamma^{\frac{2}{\max\{q,q'\}}}\right)\norm{\widetilde{T_{c_q}f}}_q \leq \norm{f}_q \leq \left(1+2^{O(d)}\gamma^{\frac{2}{\max\{q,q'\}}}\right) \norm{\widetilde{T_2f}}_q
\]
for $q'$ the H\"{o}lder conjugate of $q$. For $q\in\{4,4/3\}$, one may take $c_q=2/5$.
\end{corollary}

The proof closely follows the ideas of Bourgain, as presented by O'Donnell \cite{o2014analysis}. Based on the machinery developed in the previous section, the idea is to decompose $T_\rho$ into coordinate-wise operators and handle each coordinate as a single-variate problem, replacing each copy of $T_\rho^i$ with $T_r^i$ for $r \in \{-1,1\}$. To this end, we first need the single-variate symmetrization theorem.


\begin{lemma}[Symmetrization for random variables (10.14,10.15 in \cite{o2014analysis})]\label{lem:1-D}
Let $X$ be a $0$-mean, real-valued random variable satisfying $\norm{X}_q \leq \infty$. Then for any $a \in \mathbb{R}$, we have:
\[
\norm{a + \frac{1}{2}X}_q \leq \norm{a+rX}_q
\]
where $r \sim \{-1,1\}$ is a uniformly distributed random bit.
\end{lemma}
With this in hand, we can prove a single-coordinate variant of the result by restricting our function to the relevant variable and applying the above.
\begin{lemma}\label{lemma:coord-sym}
Let $X$ be a $d$-partite complex and $f \in C_d$. For any $i \in [d]$:
\[
\norm{T_{1/2}^if}_q \leq \norm{T^i_{r}f}_q
\]
where $r \sim \{-1,1\}$ is a uniformly distributed random bit.
\end{lemma}
\begin{proof}
\begin{align*}
    \norm{T^i_{1/2}f}_q &= \norm{\norm{T^{x_{-i}}_{1/2}f|_{x_{-i}}}_{q,x_i}}_{q,x_{-i}} & \text{(\Cref{lem:localize})}\\
    &= \norm{\norm{\left(f|_{x_{-i}}\right)^{\{\emptyset\}}+\frac{1}{2}\left(f|_{x_{-i}}\right)^{=\{i\}}}_{q,x_i,r}}_{q,x_{-i}}\\
    &\leq \norm{\norm{\left(f|_{x_{-i}}\right)^{\{\emptyset\}}+r\left(f|_{x_{-i}}\right)^{=\{i\}}}_{q,x_i,r}}_{q,x_{-i}} & \text{(\Cref{lem:1-D})}\\
    &=\norm{\norm{T_{r}f|_{x_{-i}}}_{q,x_i}}_{q,x_{-i},r}\\
    &=\norm{T^i_{r}f}_q & \text{(\Cref{lem:localize})}
\end{align*}
\end{proof}
We are now ready to prove the upper bound of \Cref{thm:symmetrization} via the replacement method.
\begin{proof}[Proof of \Cref{thm:symmetrization} (Upper Bound)]
We first argue it is sufficient to show
\begin{equation}\label{eq:noise-form}
\norm{T_{1/2}f}_q \leq \norm{T_r f}_q + 2^{O(d)}\gamma\norm{f}_q
\end{equation}
In particular we may write:
\begin{align*}
\norm{g}_q &\leq \norm{T_{1/2}T_2g}_q + 2^{O(d)}\gamma\norm{g}_q & \text{(\Cref{lemma:apx-closed})}\\
&\leq \norm{T_r T_2 g}_q + 2^{O(d)}\gamma\norm{T_2g}_q + 2^{O(d)}\gamma\norm{g}_q & \text{(\Cref{eq:noise-form})}\\
&\leq \norm{T_r T_2 g}_q + 2^{O(d)}\gamma\norm{\sum\limits_{S \subseteq [d]}2^{|S|}g^{=S}}_q + 2^{O(d)}\gamma\norm{g}_q& \text{(\Cref{claim:Efron-Noise})}\\
&\leq \norm{T_r T_2 g}_q + 2^{O(d)}\gamma\norm{g}_q & \text{(\Cref{lemma:Efron-contract})}
\end{align*}
Re-arranging and using the fact that $\frac{1}{1-2^{O(d)}\gamma} \leq 1+2^{O(d)}\gamma$ for small enough $\gamma$ gives the desired bound.

We prove \Cref{eq:noise-form} by the replacement method. In particular, we first start with de-correlating the relevant noise operators into their coordinate-wise components by \Cref{lemma:decorrelate}:
\begin{enumerate}
    \item $\norm{T_{1/2}f - T^1_{1/2}\ldots T^d_{1/2}f}_q \leq 2^{O(d)}\gamma\norm{f}_{q}.$
    \item $\norm{T_{r}f - T^1_{r_1}\ldots T^d_{r_d}f}_q \leq 2^{O(d)}\gamma\norm{f}_{q}.$
\end{enumerate}
Thus it is sufficient to prove
\[
\norm{T^1_{1/2}\ldots T^d_{1/2}f}_q \leq \norm{T^1_{r_1}\ldots T^d_{r_d}f}_q + 2^{O(d)}\gamma\norm{f}_q.
\]
We now iterate through coordinates one by one replacing $T_{1/2}$ with $T_{r_i}$. Toward this end, define $T^{(j)}$ to be the partially replaced operator at step $j$:
\[
T^{(j)} \coloneqq T^1_{r_1}\ldots T^j_{r_j}T^{j+1}_{1/2}\ldots T^d_{1/2}.
\]
We'll prove for all $j \in \{0,\ldots,d\}$:
\begin{equation}\label{eq:replacement}
    \norm{T^{(j)}f}_q \leq\norm{T^{(j+1)}f}_q + 2^{O(d)}\gamma\norm{f}_{q}.
\end{equation}
Then we clearly have
\begin{align*}
\norm{T^1_{1/2}\ldots T^d_{1/2}f}_q&=\norm{T^{(0)}f}_q\\ 
&\leq \norm{T^{(d)} f}_q + 2^{O(d)}\gamma\norm{f}_{q}\\
&= \norm{T^1_{r_1}\ldots T^d_{r_d}f}_q+2^{O(d)}\gamma\norm{f}_{q}
\end{align*}
as desired.

It remains to prove \Cref{eq:replacement}, which now follows exactly as in the standard proof up to the accumulation of error. In particular, we may simply permute the $(j+1)$st operator to the front using \Cref{lemma:decorrelate}, apply our single-coordinate symmetrization theorem to the remaining function, and permute back. Formally, observe that by \Cref{lemma:decorrelate} any two orderings $\pi$ and $\sigma$ and $r \in [-1,1]^d$ satisfy:
\begin{equation}\label{eq:permute}
    \norm{T^\pi_{r}f - T^\sigma_{r}f}_q =  \norm{(T^\pi_{r}f - T_rf) + (T_rf -T^\sigma_{r}f)}_q \leq 2^{O(d)}\gamma\norm{f}_q
\end{equation}
by the triangle inequality. Thus we can write:
\begin{align*}
    \norm{T^{(j)}f} &= \norm{T^1_{r_1}\ldots T^{j}_{r_j}T^{j+1}_{1/2}\ldots T^d_{1/2}f}_q\\
    &\leq \norm{T^{j+1}_{1/2}(T^1_{r_1}\ldots T^{j}_{r_j}T^{j+2}_{1/2}\ldots T^d_{1/2}f)}_q + 2^{O(d)}\gamma\norm{f}_q & \text{(\Cref{eq:permute})}\\
    &\leq \norm{T^{j+1}_{r_{j+1}}(T^1_{r_1}\ldots T^{j}_{r_j}T^{j+2}_{1/2}\ldots T^d_{1/2}f)}_q + 2^{O(d)}\gamma\norm{f}_q  & \text{(\Cref{lemma:coord-sym})}\\
    &\leq \norm{(T^1_{r_1}\ldots T^{j+1}_{r_{j+1}}T^{j+2}_{1/2}\ldots T^d_{1/2}f)}_q + 2^{O(d)}\gamma\norm{f}_q & \text{(\Cref{eq:permute})}\\
&=\norm{T^{(j+1)}f}_q+2^{O(d)}\gamma\norm{f}_q,
\end{align*}
completing the proof.
\end{proof}
The proof of the lower bound is similar. We start again with a closely related single-variate lemma.
\begin{lemma}[{\cite[Lemma 10.43]{o2014analysis}}]\label{lemma:sym-lower}
For any $q \geq 2$, there is an absolute constant $c_q \in (0,1)$ such that for any mean-$0$, real-valued random variable $X$ satisfying $\norm{X}_q \leq \infty$ and $a \in \mathbb{R}$:
\[
\norm{a-c_qX}_q \leq \norm{a+X}_q.
\]
For $q\in\{4,4/3\}$, we may take $c_q=2/5$.
\end{lemma}
We can now prove the lower bound in \Cref{thm:symmetrization} by similar arguments to the above.
\begin{proof}[Proof of \Cref{thm:symmetrization} (Lower Bound)]
Similar to the upper bound, it suffices to show
\[
\norm{T^{1}_{r_1c_q}\ldots T^d_{r_dc_q}f}_q \leq \norm{f}_q + 2^{O(d)}\gamma\norm{f}_q.
\]
Observe that we may re-write the lefthand side of the above inequality by breaking the norm into its boolean and product components:
\[
\norm{T^{1}_{r_1c_q}\ldots T^d_{r_dc_q}f}_q = \norm{\norm{T^{1}_{r_1c_q}\ldots T^d_{r_dc_q}f}_{q,X}}_{q,r}.
\]
It is therefore sufficient to prove that the operators $T^{i}_{r_ic_q}$ contract the inner $q$-norm for fixed $r_i \in \{\pm 1\}$. Since $0 \leq c_q \leq 1$, $T^i_{c_q}$ is an averaging operator and contracts $q$-norms by Jensen's inequality. Thus it is enough to show that for any $g$ with finite $q$-norm:
\[
\norm{T^i_{-c_q}g}_q \leq \norm{g}_q.
\]
By our localization lemma, \Cref{lem:localize}, it is sufficient to show the `single-variate' version of this contraction for $T_{-c_q}$ on any co-dimension 1 link, as then:
\[
\norm{T^i_{-c_q}g}_q = \norm{\norm{T^{x_{-i}}_{-c_q}g|_{x^{-i}}}_{q,x_i}}_{q,x_{-i}} \leq \norm{\norm{g|_{x^{-i}}}_{q,x_i}}_{q,x_{-i}} \leq \norm{g}_q.
\]
Since $T^{x_{-i}}_{-c_q}$ is self-adjoint, it is also enough to prove the inequality just for $q \geq 2$ case by H\"{o}lder conjugation (c.f.\ \Cref{lem:swap}). We can now appeal to \Cref{lemma:sym-lower}. In particular for any single variate function $h$
\begin{align*}
    \norm{T_{-c_q}h}_q &= \norm{h^{=\emptyset}-c_qh^{=1}}_q\\
    & \leq \norm{h^{=\emptyset}+h^{=1}} & \text{(\Cref{lemma:sym-lower})}\\
    &= \norm{h}_q
\end{align*}
since $h^{=1}$ is mean-$0$ and bounded. This holds in particular for $h=g|_{x^{-i}}$ on the co-dim 1 links of $X$, completing the proof.
\end{proof}
\section{Global Hypercontractivity}
We now give our first application of the symmetrization theorem, a (relatively) simple proof of nearly-sharp global hypercontractivity on HDX. 
\begin{theorem}\label{thm:Bonami}
Let $q \geq 2$ be an even integer and $X$ any $d$-partite $\gamma$-product with $\gamma \leq 2^{-\Omega(q d)}$. For any $i \leq d$ and $f \in C_d$:
\[
\norm{f^{\leq i}}_q^q \leq (500q)^{qi}\norm{f^{\leq i}}_2^2\max_{|S|\leq i, x_{S}}\{\norm{f|_{x_S}}_2^{q-2}\}
+ c_{d,q}\gamma\norm{f}_2^2\max_{|S|\leq i, x_{S}}\{\norm{f|_{x_S}}_2^{q-2}\}
\]
for $c_{d,q} \leq q^{O(qi)}2^{O(dq)}$.
\end{theorem}
Note this implies \Cref{thm:Bonami-intro} by replacing $\norm{f^{\leq i}}_2^2$ with $(1+2^{O(d)}\gamma)\norm{f}_2^2$ via approximate Parseval, and observing the error term is everywhere dominated by the main term for small enough $\gamma$.\footnote{To be fully formal, one should note the above can be proved with a constant slightly less than $500$, the error term then increases the constant to $500$ as stated in intro. Note we have made no attempt to optimize the constant.}

Before proving the result, we prove the two corollaries discussed in the introduction. Recall we call a function $f$ $r$-global if 
\[
\forall S \subset [d], x_S \in \text{Supp}(X[S]): \norm{f|_{x_S}}_2 \leq r^{|S|}\norm{f}_2.
\]
Our first corollary is hypercontractivity of the noise operator for global functions. 
\begin{corollary}\label{cor:operator-form}
        Let $q>2$ be an even integer, $r \geq 1$, and $X$ be a $d$-partite $\gamma$-product satisfying $\gamma \leq 2^{-\tilde{\Omega}(qd)}$. Then for any $\rho \leq O(\frac{1}{rq})$ and any $r$-global function $f:X(d) \to \R$
    \[
    \norm{T_\rho f}_q \leq (1+c_{d,q}(\gamma))\norm{f}_2,
    \]
    where $c_{d,q}(\gamma) \leq \tilde{O}(\frac{qd}{\log(\frac{1}{\gamma})})$ goes to $0$ as $\gamma \to 0$
\end{corollary}
\begin{proof}
    We use a variant of the standard proof of hypercontractivity from the Bonami lemma (c.f.\ \cite[Exercise 9.6]{o2014analysis}) adapted for global functions and lack of independence.

    First, observe that applying \Cref{lem:apx-eigen} and  \Cref{thm:Bonami} directly we have the weaker bound
    \begin{align*}
    \norm{T_\rho f}_q &\leq \sum\limits_{i=0}^d\rho^{i} \norm{f^{=i}}_q + 2^{O(d)}\gamma\norm{f}_2\\
    &\leq \sum\limits_{i=0}^d\rho^{i}q^{O(i)} \norm{f}_2 + 2^{O(d)}\gamma\norm{f}_2\\
    &\leq 2\norm{f}_2
    \end{align*}
whenever $\gamma \leq 2^{-\Omega(q\log(q)d)}$.\footnote{Note we have applied \Cref{thm:Bonami} to $\norm{f^{=j}}_q$ instead of $\norm{f^{\leq j}}_q$ and used the form of the bound from \Cref{thm:Bonami-intro} --- it is easily seen the statement still holds for this restricted case, or one can simply use $q$-approximate Efron Stein to move to $\norm{f^{\leq j}}_q$ first.}

To improve the constant from $2$ to $(1+o_\gamma(1))$, we will actually apply the above argument instead to the powered complex $X^{t}$ containing all faces
\[
(x_1,\ldots,x_t): x_i \in X(d)
\]
endowed with the corresponding product measure $\prod_{i=1}^t \Pi(x_i)$. Note that this is a complex of dimension $dt$. Similarly given $f: X(d) \to \R$, we define the corresponding powered function on $X^t$ as
\[
f^{\oplus t}(x_1,\ldots,x_t) = \prod\limits_{i=1}^t f(x_i).
\]
Because the $x_i$'s are independent, the standard noise operator on $X^t$ acts independently on each part. As a result it is easy to check that for any $\rho \in [0,1]$, $f: X(d) \to \R$, and $p \geq 1$:
\[
\norm{T_\rho f^{\oplus t}}_p = \norm{T_\rho f}_p^t,
\]
and moreover that if $f$ is $r$-global then $f^{\oplus t}$ is $r$-global as well. This means that we can apply the above argument for any $t$ satisfying $\gamma \leq 2^{-\Omega(q\log(q)dt)}$ to get
\[
\norm{T_\rho f}_q = \norm{T_\rho f^{\oplus t}}^{1/t}_q \leq 2^{1/t}\norm{f^{\oplus t}}_2^{1/t} \leq (1+O(1/t))\norm{f}_2.
\]
Finally setting $t = O(\frac{\log(\frac{1}{\gamma})}{q\log(q)d})$ gives the claimed result.
\end{proof}
Our second corollary is an optimal characterization of low influence functions on HDX.
\begin{corollary}\label{cor:KKL}
    Let $X$ be a $d$-partite $\gamma$-product with $\gamma \leq 2^{-\Omega(d)}$ and $f:X \to \mathbb{F}_2$ any function with influence $\mathbf{I}[f] \leq K\text{Var}(f)$. There exists $S \subset [d]$ with $|S| \leq O(K)$ and $x_S \in X[S]$ such that
    \[
    \mathbb{E}[f|_{x_S}] \geq 2^{-O(K)}.
    \]
\end{corollary}
\begin{proof}
    The proof follows from a standard `level-$i$ inequality' implied by global hypercontractivity. Namely for any boolean function $f$, we claim the degree at most $i$ Fourier mass is bounded by
    \[
    \langle f, f^{\leq i} \rangle \leq 2^{O(i)}\mathbb{E}[f]\max_{|S| \leq i, x_S}\{\mathbb{E}[f|_{x_S}]^{1/4}\}
    \]
    This follows from a basic application of H\"{o}lder's inequality:
    \begin{align*}
        \langle f, f^{\leq i} \rangle &\leq \norm{f}_{4/3}\norm{f^{\leq i}}_4 & \text{(H\"{o}lder)}\\
        &\leq 2^{O(i)}\norm{f}_{4/3}\norm{f}^{1/2}_2\max_{|S| \leq i, x_S}\{\norm{f|_{x_S}}_2^{1/2}\} & \text{(\Cref{thm:Bonami}, $q=4$)}\\
        &=  2^{O(i)}\mathbb{E}[f]\max_{|S| \leq i, x_S}\{\mathbb{E}[f|_{x_S}]^{1/4}\} & \text{(Booleanity)}
    \end{align*}
    On the other hand, recall the influence of $f$ can be written as $\mathbf{I}[f] = \sum\limits_{i=0}^d i\langle f,f^{=i}\rangle$, so in particular
    \[
    K\text{Var}(f) \geq \mathbf{I}[f] \geq (K+1)\left(\text{Var}(f) - \sum\limits_{i=1}^K\langle f, f^{=i}\rangle\right) - 2^{O(d)}\gamma\norm{f}_2^2.
    \]
    Re-arranging and using Booleanity, for small enough $\gamma$ we then have 
    \[
    \langle f, f^{\leq K}\rangle \geq \frac{1}{2}\mathbb{E}[f],
    \]
    which combined with the level-$i$ inequality implies the claimed bound.
\end{proof}
\subsection{Hypercontractivity: The Proof!}
We now turn to the proof of \Cref{thm:Bonami}.
\begin{proof}[Proof of \Cref{thm:Bonami}]
As in the proof overview, our first step is to apply the symmetrization theorem for $\gamma$-products (\Cref{cor:sym-products}) to get
\[
\norm{f^{\leq i}}^q_q \leq 2^q\norm{T_rT_2f^{\leq i}}_{q}^q.
\]
We'd like to now decompose $T_rT_2f^{\leq i}$ in terms of $f$'s Efron-Stein components. We claim this can be done up to the following loss:
\begin{claim}\label{lem:cutoff-sym}
    \[
\norm{T_rT_2f^{\leq i} - \sum\limits_{|S| \leq i}2^{|S|}f^{=S}r_S}_q \leq 2^{O(d)}\gamma\norm{f}_2^{2/q}\max_{|S| \leq i, x_S}\{\norm{f|_{x_S}}_2^{1-\frac{2}{q}}\}
\]
\end{claim}
The proof is immediate from \Cref{lem:apx-eigen} and \Cref{lemma:apx-closed}, and is deferred to the next subsection.

Using the claim and H\"{o}lder's inequality, we can now bound $\norm{f^{\leq i}}^q_q $ by
\begin{align*}
2^q\norm{T_rT_2f^{\leq i}}_{q}^q &=2^q\mathbb{E}\left[\left(\sum\limits_{|S| \leq i} 2^{|S|}r_Sf^{=S} + \left(T_rT_2f^{\leq i} - \sum\limits_{|S| \leq i} 2^{|S|}r_Sf^{=S}\right)\right)^q\right]\\
&\leq 4^q\mathbb{E}\left[\left(\sum\limits_{|S| \leq i} 2^{|S|}r_Sf^{=S}\right)^q\right] + 4^q\mathbb{E}\left[\left(T_rT_2f^{\leq i} - \sum\limits_{|S| \leq i} 2^{|S|}r_Sf^{=S}\right)^q\right]\\
&\leq 4^q\underset{x}{\mathbb{E} }\left[\underset{r}{\mathbb{E}}\left[\left(\sum\limits_{|S| \leq i}2^{|S|}f^{=S}(x)r_S\right)^q\right]\right] + 2^{O(d+q)}\gamma \norm{f}_2^2\max_{|S| \leq i, x_S}\{\norm{f|_{x_S}}_2^{q-2}\}.
\end{align*}
Notice that the inner expectation is now over a degree-$i$ boolean function with Fourier coefficients $2^{|S|}f^{=S}(x)$. Applying standard $(2{\to}q)$-hypercontractivity then gives:
\begin{align*}
    4^q\underset{x}{\mathbb{E} }\left[\underset{r}{\mathbb{E}}\left[\left(\sum\limits_{|S| \leq i}2^{|S|}f^{=S}(x)r_S\right)^q\right]\right]
    &\leq (8q)^{\frac{qi}{2}}\underset{x}{\mathbb{E}}\left[\underset{r}{\mathbb{E}}\left[\left(\sum\limits_{|S| \leq i}2^{|S|}f^{=S}(x)r_S\right)^2\right]^{q/2}\right] & \text{(Bonami Lemma)}\\
    &= (8q)^{\frac{qi}{2}}\underset{x}{\mathbb{E} }\left[\left(\sum\limits_{|S| \leq i} 2^{2|S|}f^{=S}(x)^2\right)^{q/2}\right] & \text{(Parseval)}\\
    &\leq (32q)^{\frac{qi}{2}}\underset{x}{\mathbb{E} }\left[\left(\sum\limits_{|S| \leq i} f^{=S}(x)^2\right)^{q/2}\right].
\end{align*}
We now turn our attention to the expectation:
\[
\underset{x}{\mathbb{E} }\left[\left(\sum\limits_{|S| \leq i} f^{=S}(x)^2\right)^{q/2}\right] = \underset{x}{\mathbb{E} }\left[\sum\limits_{|S_1|,\ldots,|S_{\frac{q}{2}}| \leq i} \prod\limits_{j=1}^{q/2}f^{=S_i}(x)^2\right]
\]
Our goal is eventually to reduce the above to a sum over only a single of the $q/2$ components, which we can bound by the two-norm. We will do this by iteratively `peeling off' the last $S_j$ incurring a cost in the conditional 2-norm and $q^i$ in each step. Formally, we claim 
\begin{claim}\label{claim:peel}
For any $2 \leq t \leq \frac{q}{2}$:
\begin{align*}
\underset{x}{\mathbb{E} }\left[\sum\limits_{|S_1|,\ldots,|S_{t}| \leq i} \prod\limits_{j=1}^{t}f^{=S_i}(x)^2\right] \leq& (10et)^{i}\max_{|S| \leq i,x_S}\left\{\norm{f|_{x_S}}_2^2\right\}\underset{x}{\mathbb{E} }\left[\sum\limits_{|S_1|,\ldots,|S_{t-1}| \leq i} \prod\limits_{j=1}^{t-1}f^{=S_i}(x)^2\right]\\
&+ 2^{O(dt)}\gamma\norm{f}_2^2\max_{|S| \leq i,x_S}\left\{\norm{f|_{x_S}}_2^{2t-2}\right\}
\end{align*}
\end{claim}
We now complete the proof assuming \Cref{claim:peel}. Iteratively applying the claim $\frac{q}{2}-1$ times, we arrive at
\begin{align*}
\norm{f^{\leq i}}_q^q &\leq (500q)^{qi}\max_{|S| \leq i,x_S}\left\{\norm{f|_{x_S}}_2^{q-2}\right\}\underset{x}{\mathbb{E} }\left[\sum\limits_{|S| \leq i} f^{=S}(x)^2\right] + q^{O(qi)}2^{O(dq)}\gamma \norm{f}_2^2\max_{|S| \leq i,x_S}\left\{\norm{f|_{x_S}}_2^{q-2}\right\}\\
&\leq \max_{|S| \leq i,x_S}\left\{\norm{f|_{x_S}}_2^{q-2}\right\}\norm{f^{\leq i}}_2^2 +  q^{O(qi)}2^{O(dq)}\gamma\max_{|S| \leq i,x_S}\left\{\norm{f|_{x_S}}_2^{q-2}\right\}\norm{f}_2^2
\end{align*}
as desired, where the final step is a simple application of \Cref{thm:apx-Efron-Stein}.
\end{proof}


\subsubsection{Auxiliary Proofs}
We now give the proofs of \Cref{lem:cutoff-sym} and \Cref{claim:peel}.
\begin{proof}[Proof of \Cref{lem:cutoff-sym}]

    The proof is immediate from linearity of the noise operator and the triangle inequality if we can show for any $|S| \leq i$:
    \[
    \norm{T_rT_2f^{=S} - 2^{|S|}r_Sf^{=S}}_q \leq 2^{O(d)}\gamma\norm{f}_2^{2/q}\max_{|S| \leq i, x_S}\{\norm{f|_{x_S}}_2^{1-\frac{2}{q}}\}.
    \]
    By \Cref{claim:Efron-Noise} we can write
    \begin{align*}
        T_rT_2f^{=S} &= \sum\limits_{S' \subseteq [d]} 2^{|S|}T_r(f^{=S})^{=S'}.
    \end{align*}
    Since the projection operators $\{E_{T}\}$ contract $p$-norms and $T_\rho$ is a (bounded) linear combination of projections, \Cref{lemma:apx-closed} implies the righthand side is $2^{O(d)}\gamma\norm{f}_2^{2/q}\max_{|S| \leq i, x_S}\{\norm{f|_{x_S}}_2^{1-\frac{2}{q}}\}$-close in $q$-norm to $\sum_S 2^{|S|}T_rf^{=S}$. Expanding out $T_r$, the main term is then
    \[
    2^{|S|}\sum\limits_{T \subseteq [d]}r_T\prod_{i \notin T}(1-r_i) E_Tf^{=S}.
    \]
    Now by \Cref{lemma:p-to-q-down}, $\norm{E_Tf^{=S}}_q \leq 2^{O(d)}\gamma\norm{f}_2^{2/q}\max_{|S| \leq i, x_S}\{\norm{f|_{x_S}}_2^{1-\frac{2}{q}}\}$ unless $S \subset T$, in which case $E_Tf^{=S} = f^{=S}$. Thus the above is similarly close in $q$-norm to
    \begin{align*}
    2^{|S|}\sum\limits_{T \supset S}r_T\prod_{i \notin T}(1-r_i)f^{=S} &= 2^{|S|}r_Sf^{=S}
    \end{align*}
    as desired (where the equality follows as in \Cref{claim:Efron-Noise}).    
\end{proof}
\begin{proof}[Proof of \Cref{claim:peel}] Recall our goal is to show for any $2 \leq t \leq \frac{q}{2}$:
\begin{align*}
\underset{x}{\mathbb{E} }\left[\sum\limits_{|S_1|,\ldots,|S_{t}| \leq i} \prod\limits_{j=1}^{t}f^{=S_i}(x)^2\right] \leq& (10et)^i \max_{|S| \leq i,x_S}\left\{\norm{f|_{x_S}}_2^2\right\}\underset{x}{\mathbb{E} }\left[\sum\limits_{|S_1|,\ldots,|S_{t-1}| \leq i} \prod\limits_{j=1}^{t-1}f^{=S_i}(x)^2\right]\\
&+ 2^{O(d)}\gamma\norm{f}_2^2\max_{|S| \leq i,x_S}\left\{\norm{f|_{x_S}}_2^{2t-2}\right\}
\end{align*}
Our proof strategy follows an extension of the $q=4$ case (the $t=2$ setting of this lemma) proved for products in \cite{zhao2021generalizations} combined with variants of our $q$-norm expansion bounds to handle correlated variables.

In particular, our goal is now to pull out all $S_t$ terms in the sum only paying a factor of $t^{i}$ and the conditional $2$-norm of $f$. We do this by re-indexing our sum over \textit{intersections} so that the $S_t$ terms (conditioned on the intersection) are ``independent" of $S_1,\ldots,S_{t-1}$. We then pay only a factor in the conditional norm when pulling out the maximum.

Formally, write $\mathcal{S}^{t-1}=\bigcup_{j=1}^{t-1} S_{j}$ and the set of intersecting coordinates as
\[
I=\mathcal{S}^{t-1} \cap S_t.
\]
We can upper bound the left-hand expectation by first summing over intersections $I$:
\begin{equation}\label{eq:hyp-formal}
\sum\limits_{|I| \leq i}\underset{x_{I}}{\mathbb{E} }\left[\sum\limits_{\underset{\mathcal{S}^{t-1} \supset I}{|S_1|,\ldots,|S_{t-1}| \leq i:}} \underset{x_{\mathcal{S}^{t-1} \setminus I}}{\mathbb{E}}\left[\prod\limits_{j=1}^{t-1}f^{=S_i}(x)^2\cdot \left(\sum\limits_{\underset{S_t \cap \mathcal{S}^{t-1}=I}{S_t \supset I:}}\underset{x_{S_t \setminus I}}{\mathbb{E}}\left[f^{=S_t}|_{x_{I}}(x_{S_t})^2\right]\right)\right] \right]
\end{equation}
Note that here the inner expectations are drawn \textit{conditionally} on the prior variables. This would not matter in a product, as the variables would be independent, but here our inner expectation over $x_{S_t \setminus I}$ actually depends on variables in $S_1,\ldots,S_{t-1}$ that are outside of $S_t$. This is an issue for us since we want to relate the maximum of this final term to $\max_{I,x_I}\{\norm{f|_{x_I}}\}$ and it is not clear how to do this under such dependencies.

To handle this fact we will draw on a variant of an idea in \cite{bafna2021hypercontractivity}'s proof of hypercontractivity for two-sided HDX. Namely, it is possible to use expansion of swap walks to decorrelate these dependencies:
\begin{claim}\label{claim:de-cor-hyp}
    \begin{align*}
    (\ref{eq:hyp-formal})
    \leq &\sum\limits_{|I| \leq i}\underset{x_{I}}{\mathbb{E} }\left[\sum\limits_{\underset{\mathcal{S}^{t-1} \supset I}{|S_1|,\ldots,|S_{t}| \leq i:}}\underset{x_{\mathcal{S}^{t-1} \setminus I}}{\mathbb{E}}\left[ \prod\limits_{j=1}^{t-1}f^{=S_i}(x)^2\right]\cdot \left(\sum\limits_{S_t \supset I}\underset{x_{S_t \setminus I}}{\mathbb{E}}\left[f^{=S_t}|_{x_{I}}(x_{S_t})^2\right]\right)\right]\\
    &+ 2^{O(d)}\gamma\norm{f}_2^2\max_{|S| \leq i,x_S}\left\{\norm{f|_{x_S}}_2^{2t-2}\right\}.
    \end{align*}
\end{claim}

We defer the proof and complete the argument. Pulling out the maximum we can bound the main term by
\[
\max_{|I| \leq i, x_I}\left(\sum\limits_{S_t \supset I}\underset{x_{S_t \setminus I}}{\mathbb{E}}\left[f^{=S_t}|_{x_{I}}(x_{S_t})^2\right]\right)\sum\limits_{|I| \leq i}\underset{x_{I}}{\mathbb{E} }\left[\sum\limits_{\underset{\mathcal{S}^{t-1} \supset I}{|S_1|,\ldots,|S_{t}| \leq i:}}\underset{x_{\mathcal{S}^{t-1} \setminus I}}{\mathbb{E}}\left[ \prod\limits_{j=1}^{t-1}f^{=S_i}(x)^2\right]\right]
\]

We bound the two terms separately. For the latter, observe that every term in the inner summation occurs at most $\sum_{j \leq i}{i(t-1) \choose j} \leq (2et)^{i}$ times (once for each possible choice of intersecting subset $I$), so we have
\[
\sum\limits_{|I| \leq i}\underset{x_{I}}{\mathbb{E} }\left[\sum\limits_{\underset{\mathcal{S}^{t-1} \supset I}{|S_1|,\ldots,|S_{t-1}| \leq i:}} \underset{x_{\mathcal{S}^{t-1} \setminus I}}{\mathbb{E}}\left[\prod\limits_{j=1}^{t-1}f^{=S_i}(x)^2\right]\right] \leq (2et)^i \underset{x}{\mathbb{E} }\left[\sum\limits_{|S_1|,\ldots,|S_{t-1}| \leq i} \prod\limits_{j=1}^{t-1}f^{=S_i}(x)^2\right].
\]

To bound the maximum, recall from \Cref{lem:Efron-Restrict} that
\[
f^{=I \cup B}(y_I,x_B) = \sum\limits_{J \subset I} (-1)^{|I|-|J|}(f|_{y_J})^{=B}(x_B).
\]
Thus for any possible $y_I \in X[I]$, a simple application of Cauchy-Schwarz gives:
\begin{align*}
\sum\limits_{T \supset I}\underset{x_{T \setminus I} \sim X_{y_I}}{\mathbb{E}}\left[f^{=T}(y_I,x_{T\setminus I})^2\right] &= \sum\limits_{T \supset I}\underset{x_{T \setminus I} \sim X_{y_I}}{\mathbb{E}}\left[\left(\sum\limits_{J \subset I}(-1)^{|I|-|J|}(f|_{y_J})^{=T\setminus I}(x_{T\setminus I})\right)^2\right]\\
&\leq 2^{i}\sum\limits_{J \subset I}\sum\limits_{S \subset \bar{I}} \underset{x_S}{\mathbb{E}}[(f|_{y_J})^{=S}(x_S)^2]\\
    &\leq 2^{i}\sum\limits_{J \subset I}\sum\limits_{S \subset \bar{J}} \underset{x_S}{\mathbb{E}}[(f|_{y_J})^{=S}(x_S)^2]\\
    &\leq 5^i\max_{|J| \leq |I|,y_J}\{\norm{f|_{y_J}}_2^2\}
\end{align*}
where in the last step we have applied approximate Parseval (\Cref{thm:apx-Efron-Stein}) within the link of $y_J$ (hence the $5^i$ rather than $4^i$ coefficient).
\end{proof}
\begin{proof}[Proof of \Cref{claim:de-cor-hyp}]
    We have now reached the final step, decorrelating the variables in $S_t \setminus I$ in our inner expectation from the remaining `irrelevant' variables in $S_1,\ldots,S_t \setminus I$. The key is to observe that we can rewrite the inner expectation as an application of the swap walk between these disjoint variable sets in the link of $x_I$:
\[
\eqref{eq:hyp-formal} = \sum\limits_{|I| \leq i}\underset{x_{I}}{\mathbb{E} }\left[\sum\limits_{\underset{\mathcal{S}^{t-1} \supset I}{|S_1|,\ldots,|S_{t-1}| \leq i:}}\underset{x_{\mathcal{S}^{t-1} \setminus I}}{\mathbb{E}}\left[ \prod\limits_{j=1}^{t-1}f^{=S_i}(x)^2\cdot \left(\sum\limits_{\underset{S_t \cap \mathcal{S}^{t-1}=I}{S_t \supset I:}}A^{x_I}_{S_t \setminus I, \mathcal{S}^{t-1} \setminus I}(f^{=S_t}|_{x_{I}})^2\right)\right] \right]
\]
Since the swap walks are $d\gamma$-expanders \cite{gur2021hypercontractivity}, we can re-write the above as the desired term de-correlating $S_t$ as in the claim, plus an error term of the form
\[
\sum\limits_{|I| \leq i}\underset{x_{I}}{\mathbb{E} }\left[\sum\limits_{\underset{\mathcal{S}^{t-1} \supset I}{|S_1|,\ldots,|S_{t-1}| \leq i:}} \underset{x_{\mathcal{S}^{t-1} \setminus I}}{\mathbb{E}}\left[\prod\limits_{j=1}^{t-1}f^{=S_i}(x)^2\cdot \left(\sum\limits_{\underset{S_t \cap \mathcal{S}^{t-1}=I}{S_t \supset I:}}\Gamma(f^{=S_t}|_{x_{I}})^2\right)\right] \right]
\]
where $\Gamma$ has operator norm at most $d\gamma$. Applying Cauchy-Schwarz individually to each summand, we can then upper bound this by
\[
d\gamma\sum\limits_{|I| \leq i}\underset{x_{I}}{\mathbb{E} }\left[\sum\limits_{\underset{\mathcal{S}^{t-1} \supset I}{|S_1|,\ldots,|S_{t-1}| \leq i:}} \underset{x_{\mathcal{S}^{t-1} \setminus I}}{\mathbb{E}}\left[\prod\limits_{j=1}^{t-1}f^{=S_i}(x)^4\right]^{1/2}\cdot \left(\sum\limits_{\underset{S_t \cap \mathcal{S}^{t-1}=I}{S_t \supset I:}}\underset{x_{S_t \setminus I}}{\mathbb{E}}\left[f^{=S_t}(x)^4\right]^{1/2}\right) \right].
\]
Now repeatedly applying Cauchy-Schwarz within the product $t-2$ additional times (separating out only the final $S_j$ in each step), we can loosely upper bound this cumbersome expression by the much nicer one

\[
d\gamma\sum\limits_{S_1,\ldots,S_t}\norm{f^{=S_t}}_{2^{t}}^2\prod_{j=1}^{t-1}\norm{f^{=S_j}}_{2^{j+1}}^2,
\]
where we have removed the restrictions on the intersection and size of the $S_i$, and reversed the order of the $S_i$ to simplify the expression (note the former is valid since we have only added to the value).

Finally we apply \Cref{lemma:Efron-contract} to each term in the product and simplify to get
\begin{align*} &2^{O(ti)}d\gamma\sum\limits_{S_1,\ldots,S_t}\norm{f}_{2}^{2^{-t}}\max_{|I| \leq i, x_I}\{\norm{f|_{x_I}}_{2}^{2-2^{-t}}\}\prod_{j=1}^{t-1}\norm{f}_{2}^{2^{-j+1}}\max_{|I| \leq i, x_I}\{\norm{f|_{x_I}}_{2}^{2-2^{-j+1}}\}\\
\leq &2^{O(dt)}\gamma\norm{f}_2^2\max_{|I| \leq i, x_I}\{\norm{f|_{x_I}}_{2}^{2t-2}\}
\end{align*}
as desired.
\end{proof}
\section{Bourgain's Booster Theorem}
In this section we prove $\gamma$-products satisfy a booster theorem, resolving a main open question of \cite{bafna2021hypercontractivity}. We first define the relevant notion of a \textit{booster}, which is simply a restriction on which $f$ deviates significantly from its expectation.
\begin{definition}[$\tau$-boosters]
    Given $f: X(d) \to \{\pm 1\}$, a sub-assignment $x_T \in X[T]$ is called a $\tau$-booster if 
    \[
    |\mathbb{E}[f|_{x_T}]-\mathbb{E}[f]| \geq \tau.
    \]
\end{definition}
We prove that low influence functions on HDX have many (small support) boosters.
\begin{theorem}\label{thm:booster}
    Let $X$ be a $\gamma$-product for $\gamma \leq 2^{-\Omega(d)}$ and $f: X(d) \to \{\pm 1\}$ a function with $\mathbf{I}[f] \leq K$. If $\text{Var}(f) \geq .01$, there is some $\tau \geq 2^{-O(K^2)}$ such that
    \[
    \Pr_{x \sim X(d)}\left[\exists T \subset [d]:~|T| \leq O(K) \text{ and $x_T$ is a $\tau$-booster} \right] \geq \tau
    \]
\end{theorem}
Our proof closely follows O'Donnell's treatment of Bourgain's Theorem for product spaces in \cite{o2014analysis}, adjusting where necessary to handle lack of independence. The key technical component is to show that for any low influence function, it is possible to identify a small \textit{input-dependent} set of coordinates which account for most of the Fourier mass.
\begin{proposition}\label{prop:booster}
     Let $X$ be a $\gamma$-product with $\gamma \leq 2^{-\Omega(d)}$, $\varepsilon \in (2^{O(d)}\gamma, 1/2)$, and $f:X(d) \to \{\pm 1\}$ any function satisfying $\text{Var}(f) \geq .01$. Let $\ell=\mathbf{I}[f]/\varepsilon$. There exists a family of `notable coordinates' $\{J_x\}_{x \in X(d)}$ such that
     \begin{enumerate}
         \item $J_x$ is small: 
         \[
         \forall x, |J_x| \leq 2^{O(\ell)}
         \]
         \item Small subsets of $J_x$ contain of most of the Fourier mass:
         \[
         \mathbb{E}\left[\sum\limits_{S \notin \mathcal{F}_x} f^{=S}(x)^2\right] \leq 3\varepsilon
         \]
     \end{enumerate}
     where $\mathcal{F}_x \coloneqq \{S: S \subseteq J_x, |S| \leq \ell\}$.
\end{proposition}
We first prove the booster theorem assuming this fact. The proof is essentially exactly as in the product case (see e.g.\ \cite[Page 311]{o2014analysis}). We include the argument for completeness.
\begin{proof}[Proof of \Cref{thm:booster}]
    We first argue that it is enough to show that with good probability over $x$ there is a subset of coordinates $0 < |S_x| \leq O(K)$ such that $f^{=S_x}(x)^2 \geq 2^{-O(K^2)}$, and in particular that this implies there exists some $T \subseteq S_x$ such that $x_T$ is a $2^{-O(K^2)}$-booster.

    Toward this end, let $g=f-\mathbb{E}[f]$, and note for all $S \neq \emptyset$, $g^{=S}=f^{=S}$. Since $S_x \neq \emptyset$, we therefore have $|g^{=S_x}(x)| \geq 2^{-O(K^2)}$. Finally, recall
    \[
    g^{=S_x}(x) = \sum\limits_{\emptyset \neq T \subseteq S_x}(-1)^{|S_x|-|T|}E_Tg(x),
    \]
    so there must exist some $0 < T \leq O(K)$ such that $|E_Tg| \geq 2^{-O(K^2)}$. However since $g=f-\mathbb{E}[f]$, we have
    \[
    |E_Tg(x)| = |\mathbb{E}[f|_{x_T}]-\mathbb{E}[f]| \geq 2^{-O(K^2)}
    \]
    as desired.

    With this in mind, our goal is now to define such sets $S_x$ so that
    \[
    \Pr_x[f^{=S_x}(x)^2 \geq 2^{-O(K^2)}] \geq 2^{-O(K^2)}.
    \]
    Set $\varepsilon$ of \Cref{prop:booster} to $.001$ (note this is admissible when $\gamma \leq 2^{-\Theta(d)}$ is small enough), and for any fixed $x \in X(d)$ let $|\mathcal{F}_x| \leq 2^{O(K^2)}$ be the resulting family of subsets. We will show that there is often a subset $S_x \in \mathcal{F}_x$ with large Fourier weight.

    First, we observe that the \textit{expected} weight of such $f^{=S_x}(x)^2$ drawn from $\mathcal{F}_x$ is large. Namely because $\text{Var}(f) \geq .01$, we have $f^{=\emptyset}(x)^2 \leq .99$ and
    \[
    \mathbb{E}_x\left[\sum\limits_{S \in \mathcal{F}_x \setminus \{\emptyset\}}f^{=S}(x)^2\right] \geq 1-3\varepsilon-.99-2^{O(d)}\gamma \geq .005.
    \]
    This in turn implies that in expectation over $x$, the \textit{maximum} weight coefficient in $ \mathcal{F}_x$ is at least $2^{-O(K^2)}$:
    \[
    \mathbb{E}_x\left[\max_{S \in \mathcal{F}_x \setminus \{\emptyset\}}f^{=S}(x)^2\right] \geq 2^{-O(K^2)}.
    \]
    In particular, this means for every $x$ we can define a set $0< |S_x| \leq O(K)$ such that
    \[
    \mathbb{E}_x\left[f^{=S_x}(x)^2\right] \geq 2^{-O(K^2)}.
    \]
    Finally, recalling $\norm{f^{=S_x}}_\infty \leq 2^{|S_x|}$ (\Cref{lemma:Efron-contract}) we infer
    \[
    \Pr_x[f^{=S_x}(x)^2 \geq 2^{-O(K^2)}] \geq 2^{-O(K^2)}
    \]
    as desired.
\end{proof}
It is left to prove the core Proposition.
\begin{proof}[Proof of \Cref{prop:booster}]
    First observe that the expected Fourier mass on components beyond level $\ell$ is small
    \[
    \mathbb{E}_x\left[\sum\limits_{|S| > \ell} f^{=S}(x)^2 \right] \leq 2\varepsilon.
    \]
    This follows by viewing the components $\sum\limits_{|S|=i}\mathbb{E}_x[f^{=S}(x)^2]$ as an approximate distribution over $[d]$. Namely using the fact that
    \[
    1=\mathbb{E}_x[f^2] \in (1\pm 2^{O(d)}\gamma)\sum\limits_{S \subseteq [d]}\mathbb{E}_x[f^{=S}(x)^2],
    \]
    for small enough $\gamma$, there is some normalizing factor $c \in [.9,1.1]$ such that $c\mathbb{E}_x[f^{=S}(x)^2]$ is a distribution. Now consider the random variable $Z$ that takes value $|S|$ with probability $c\mathbb{E}_x[f^{=S}(x)^2]$. By Markov:
    \[
    \Pr[Z > \ell] \leq \frac{\mathbb{E}[Z]}{\ell} \leq \frac{c^{-1}\mathbf{I}[f]+c_d\gamma}{\ell} \leq c^{-1}\varepsilon+\frac{\varepsilon c_{d}\gamma}{\mathbf{I}[f]}
    \]
    To bound the latter `error' term, note that by approximate orthogonality of Efron-Stein, we have the relation
    \[
    \mathbf{I}[f] = \sum\limits_{i=0}^d i\langle f,f^{=i} \rangle \geq \text{Var}(f) - 2^{O(d)}\gamma \geq .005,
    \]
    which combined with the above gives
    \[
    \Pr[Z > \ell] \leq 2\varepsilon
    \]
    as desired.
    
    We may now restrict our attention to the components of degree at most $\ell$. In particular, we need to find small sets $J_x \subset [d]$ such that
    \[
    \mathbb{E}_x\left[\sum\limits_{|S| \leq \ell, S \not\subseteq J_x}f^{=S}(x)^2\right] \leq \varepsilon.
    \]
    The natural strategy to define such a set is simply to take any coordinate with large influence, that is
    \[
    J_x' \coloneqq \left\{j \in [d]:~\sum\limits_{S \ni j}f^{=S}(x)^2 \geq \tau \right\}
    \]
    where $\tau = 2^{-\Theta(\ell)}$ is a sufficiently small constant. In fact we will show for this definition
    \begin{equation}\label{eq:booster-1}
       \mathbb{E}_x\left[\sum\limits_{|S| \leq \ell, S \not\subseteq J_x'}f^{=S}(x)^2\right] \leq \varepsilon/2
    \end{equation}
    The issue is that $|J_x'|$ may not be bounded. We will argue that truncating $J_x'$ as
    \[
    J_x =
    \begin{cases}
        J_x' & \text{if $|J_x'| \leq C^\ell$}\\
        \emptyset & \text{otherwise}
    \end{cases}
    \]
    for some large enough constant $C>0$ gives the desired family. In particular, it is enough to show
    \begin{equation}\label{eq:booster-2}
    \mathbb{E}\left[\mathbf{1}[|J_x'| > C^\ell]\sum\limits_{0 < |S| \leq \ell} f^{=S}(x)^2\right] \leq \varepsilon/2.
    \end{equation}
We first show \Cref{eq:booster-1}. Toward this end, we will start by adding two applications of the noise operator inside our summation, one in order to use symmetrization, the second to use boolean hypercontractivity inside the symmetrized function. In particular, using the fact that Efron-Stein is an approximate eigendecomposition (\Cref{lem:apx-eigen}) and $\norm{f}_2=1$, we can write:

\begin{align*}
     \mathbb{E}_x\left[\sum\limits_{|S| \leq \ell, S \not\subseteq J_x'}f^{=S}(x)^2\right] & \leq 2^{O(\ell)}\mathbb{E}_x\left[\sum\limits_{S \not\subseteq J_x'}(T_{1/\sqrt{3}}T_{2/5}f^{=S})(x)^2\right] + 2^{O(d)}\gamma
\end{align*}
We'd like to eventually compare the inside sum to the total influence of $f$. To do this, we'll re-index our sum over coordinates not $i \notin J'_x$, then over $S \ni i$. We can then view the inner sum as an application of the Laplacian. In particular continuing the above we have:
\begin{align*}
     & \leq 2^{O(\ell)}\mathbb{E}_x\left[\sum\limits_{i \not\in J_x'}\sum\limits_{S \ni i}(T_{1/\sqrt{3}}T_{2/5}f^{=S})(x)^2\right] + 2^{O(d)}\gamma\\
     & \leq 2^{O(\ell)}\mathbb{E}_x\left[\sum\limits_{i \not\in J_x'}(T_{1/\sqrt{3}}T_{2/5}L_if)(x)^2\right] + 2^{O(d)}\gamma
\end{align*}
where the last step follows from observing
\begin{align*}
    \mathbb{E}_x\left[\sum\limits_{S \ni i}(T_{1/\sqrt{3}}T_{2/5}f^{=S})(x)^2\right] &= \mathbb{E}_x\left[(T_{1/\sqrt{3}}T_{2/5}L_if)(x)^2\right] - \sum\limits_{S \neq S' \ni i} \langle T_{1/\sqrt{3}}T_{2/5}f^{=S},T_{1/\sqrt{3}}T_{2/5}f^{=S'} \rangle\\
    &\leq 2^{O(d)}\gamma
\end{align*}
by \Cref{lem:apx-eigen} and approximate orthogonality.

Now defining $g^{i} \coloneqq T_{2/5}L_if$, we can re-write the above as
\[
\mathbb{E}_x\left[\sum\limits_{|S| \leq \ell, S \not\subseteq J_x'}f^{=S}(x)^2\right] \leq 2^{O(\ell)}\sum\limits_{i \not\in J_x'}\norm{(T_{1/\sqrt{3}}g^i)}_2^2 + 2^{O(d)}\gamma.
\]
We'll now pass to the symmetrization of $g^i$ so we can apply boolean hypercontractivity. In fact, we argue this can be done directly, namely:
\begin{align*}
2^{O(\ell)}\sum\limits_{i \not\in J_x'}\norm{(T_{1/\sqrt{3}}g^i)}_2^2 + 2^{O(d)}\gamma \leq 2^{O(\ell)}\sum\limits_{i \not\in J_x'}\mathbb{E}_x\left[\norm{T_{1/\sqrt{3}}\widetilde{g^i}|_x}_{2,r}^2\right] + 2^{O(d)}\gamma
\end{align*}
This follows since $\widetilde{g^i}|_x$ is the boolean function whose Fourier coefficients are $(g^i)^{=S}(x)$, so by Parseval
\begin{align*}
\mathbb{E}_x\left[\norm{T_{1/\sqrt{3}}\widetilde{g^i}|_x}_{2,r}^2\right] &= \mathbb{E}_x\left[\sum\limits_{S \subseteq [d]}\left(\frac{1}{\sqrt{3}}\right)^{|S|}(g^i)^{=S}(x)^2\right]\\
&\geq \mathbb{E}_x\left[T_{1/\sqrt{3}}g^i(x)^2\right] - 2^{O(d)}\gamma
\end{align*}
where we've again used approximate orthogonality and \Cref{lem:apx-eigen}.

We are now finally in position to apply boolean hypercontractivity to $\widetilde{g^i}|_x$. We will do this indirectly via the following useful lemma adapted from \cite[Lemma 10.48]{o2014analysis}.
\begin{lemma}\label{lemma:2-vs-43}
    Let $i \notin J_x'$. Then
    \[
    \mathbb{E}_x\left[\norm{T_{\frac{1}{\sqrt{3}}}\widetilde{g^i}|_x}_2^2\right] \leq 2\tau^{1/3}\mathbb{E}_x\left[\norm{\widetilde{g^i}|_x}_{4/3}^{4/3}\right]+2^{O(d)}\gamma^{1/3}
    \]
\end{lemma}
Combining this with our prior arguments, we can bound the total Fourier weight outside $J'_x$ as:
\begin{align*}
    \mathbb{E}_x\left[\sum\limits_{|S| \leq \ell, S \not\subseteq J_x'}f^{=S}(x)^2\right] &\leq 2^{O(\ell)}\sum\limits_{i \not\in J_x'}\mathbb{E}_x\left[\norm{T_{1/\sqrt{3}}\widetilde{g^i}|_x}_{2,r}^2\right] + 2^{O(d)}\gamma\\
    &\leq 2^{O(\ell)}\tau^{1/3}\sum\limits_{i \not\in J_x'}\mathbb{E}_x\left[\norm{\widetilde{T_{2/5}L_if}|_x}_{4/3,r}^{4/3}\right]+2^{O(d)}\gamma^{1/3}\\
    &\leq 2^{O(\ell)}\tau^{1/3}\sum\limits_{i \in [d]}\norm{L_if}_{4/3}^{4/3}+2^{O(d)}\gamma^{1/3} & \text{(Symmetrization Theorem)}
\end{align*}
We are left with needing to relate the $4/3$-norm of the Laplacian applied to $f$ to its total influence. This is a standard exercise, indeed for any $\{\pm 1\}$-valued functions we have $\norm{L_if}_{4/3}^{4/3} \leq O(\langle f,L_i f \rangle)$.\footnote{See e.g.\ \cite[Exercise 8.10]{o2014analysis}. This holds on any partite complex by observing $\norm{L_if}_{4/3}^{4/3}=\mathbb{E}[|L_if(x)|^{4/3}] \leq 2^{1/3}\mathbb{E}[|L_if(x)|] = 2^{1/3}\mathbb{E}[f(x)\cdot L_if(x)] = 2^{1/3}\langle f, L_i f \rangle$ where the 3rd and 4th step use that $f$ is $\{\pm 1\}$-valued.} Thus we can finally upper bound this by
\[
2^{O(\ell)}\tau^{1/3}\mathbf{I}[f] + 2^{O(d)}\gamma^{1/3} \leq \varepsilon/2
\]
for the appropriate choice of $\tau$ as desired.

We now need to show that $|J'_x|$ is typically small (\Cref{eq:booster-2}). We first separate the indicator and Fourier sum by Cauchy-Schwarz:
\begin{align*}
     \mathbb{E}\left[\mathbf{1}[|J_x'| > C^\ell]\sum\limits_{0 < |S| \leq \ell} f^{=S}(x)^2\right] \leq \Pr_x[|J'_x| > C^\ell]^{1/2}\cdot  \mathbb{E}\left[\left(\sum\limits_{0 < |S| \leq \ell} f^{=S}(x)^2\right)^2\right]^{1/2}
\end{align*}
We bound the two terms separately. Toward the first, observe that by definition for any fixed $x$ the size of $|J'_x|$ can be at most
\[
|J'_x| \leq \frac{1}{\tau}\sum\limits_{i \in [d]}\sum\limits_{S \ni i}f^{=S}(x)^2,
\]
else the mass in $J'_x$ exceeds the total mass at $x$. With this in mind, we can bound $\Pr_x[|J'_x| > C^\ell]$ by Markov's inequality as:
\begin{align*}
    \Pr_x[|J'_x| > C^\ell] &\leq C^{-\ell}\mathbb{E}_x[|J'_x|]\\
    &\leq C^{-\ell}\mathbb{E}_x[\frac{1}{\tau}\sum\limits_{i \in [d]}\sum\limits_{S \ni i}f^{=S}(x)^2]\\
    &\leq \frac{C^{-\ell}}{\tau}(\mathbf{I}[f]+2^{O(d)}\gamma)
\end{align*}
where we've used the relation $\mathbf{I}[f] = \sum\limits_{i \leq d}i\langle f,\sum\limits_{S \ni i} f^{=S} \rangle$ (\Cref{claim:total-inf}) and approximate orthogonality.


To bound the latter term, let $h = T_{2/5}(f-f^{=\emptyset})$. By similar arguments as above, we may write
\begin{align*}
    \mathbb{E}_x\left[\left(\sum\limits_{0 < |S| \leq \ell} f^{=S}(x)^2\right)^2\right] &\leq 2^{O(\ell)}\mathbb{E}\left[\left(\sum\limits_{S \neq \emptyset} T_{2/5}f^{=S}(x)^2\right)^2\right] + 2^{O(d)}\gamma\\
    &\leq 2^{O(\ell)}\mathbb{E}_x\left[\norm{\tilde{h}|_x}_{2,r}^4\right] + 2^{O(d)}\gamma\\
    &\leq 2^{O(\ell)}\mathbb{E}_x\left[\norm{\tilde{h}|_x}_{4,r}^4\right] + 2^{O(d)}\gamma\\
    &\leq 2^{O(\ell)}\norm{f-f^{\emptyset}}_{4}^4 + 2^{O(d)}\gamma & \text{(Symmetrization)}\\
    &\leq 2^{O(\ell)}\text{Var}(f) + 2^{O(d)}\gamma\\
    &\leq 2^{O(\ell)}\mathbf{I}[f] + 2^{O(d)}\gamma
\end{align*}
where we've used the fact that $|f-f^{\emptyset}| \leq 2$ to move to the variance and again applied the inequality $\text{Var}(f) \leq \mathbf{I}[f] - 2^{O(d)}\gamma$.

Combining these we have
\[
\mathbb{E}\left[\mathbf{1}[|J_x'| > C^\ell]\sum\limits_{0 < |S| \leq \ell} f^{=S}(x)^2\right] \leq (2^{O(\ell)}\mathbf{I}[f] + 2^{O(d)}\gamma)^{1/2}\left(\frac{C^{-\ell}}{\tau}(\mathbf{I}[f]+2^{O(d)}\gamma)\right)^{1/2} \leq \varepsilon/2
\]
for large enough constant $C$ and small enough $\gamma$.
\end{proof}

It is left to prove \Cref{lemma:2-vs-43}.
\begin{proof}[Proof of \Cref{lemma:2-vs-43}]
    Fix any $x$. Inside the expectation we may appeal to standard $(4/3,2)$-hypercontractivity which states that $\norm{T_{1/\sqrt{3}}\widetilde{g^i}|_x}_2 \leq \norm{\widetilde{g^i}|_x}_{4/3}$. Thus
    \[
    \norm{T_{1/\sqrt{3}}\widetilde{g^i}|_x}_2^2 \leq \left(\norm{\widetilde{g^i}|_x}_{4/3}^{2}\right)^{1/3} \cdot \norm{\widetilde{g^i}|_x}_{4/3}^{4/3} \leq \left(\norm{\widetilde{g^i}|_x}_{2}^{2}\right)^{1/3} \cdot \norm{\widetilde{g^i}|_x}_{4/3}^{4/3} 
    \]
    We claim that the $2$-norm term on the RHS can be bounded as a function of $x$ by:
    \begin{equation}\label{eq:g-sym-norm-bound}
    \norm{\widetilde{g^i}|_x}_{2}^{2} \leq \tau + \Gamma(x)
    \end{equation}
    where $\norm{\Gamma}_2 \leq 2^{O(d)}\gamma$ and $\Gamma$ is non-negative. Assuming this fact, we can now bound
    \begin{align*}
\mathbb{E}_x\left[\norm{T_{1/\sqrt{3}}\widetilde{g^i}|_x}_2^2\right] &\leq \mathbb{E}_x\left[(\tau + \Gamma)^{1/3} \cdot \norm{\widetilde{g^i}|_x}_{4/3}^{4/3}\right]\\
& \leq \tau^{1/3}\mathbb{E}_x\left[ \norm{\widetilde{g^i}|_x}_{4/3}^{4/3}\right] + 2^{O(d)}\norm{\Gamma}_{1/3}^{1/3}\\
& \leq \tau^{1/3}\mathbb{E}_x\left[ \norm{\widetilde{g^i}|_x}_{4/3}^{4/3}\right] + 2^{O(d)}\norm{\Gamma}_{2}^{1/3}\\
&\leq \tau^{1/3}\mathbb{E}_x\left[ \norm{\widetilde{g^i}|_x}_{4/3}^{4/3}\right] + 2^{O(d)}\gamma^{1/3}
    \end{align*}
where we've used the fact that $\max_x \norm{\widetilde{g^i}|_x}_{4/3}^{4/3} \leq 2^{O(d)}$ and monotonicity of $\ell_p$-norms.

To prove \Cref{eq:g-sym-norm-bound}, observe by standard Parseval we have
    \begin{align*}
    \norm{\tilde{g}|_x}_{2}^{2} &= \sum\limits_{S \subseteq [d]} (T_{2/5}L_if)^{=S}(x)^2\\
    &=\sum\limits_{S \subseteq [d]} \sum\limits_{T \ni i}(T_{2/5}f^{=T})^{=S}(x)^2
    \end{align*}
    By \Cref{lem:apx-eigen} and \Cref{lemma:apx-closed}, each individual term $(T_{2/5}f^{=T})^{=S}$ is $2^{O(d)}\gamma$-close in $2$-norm to
    \[
    \begin{cases}
        (\frac{2}{5})^{|S|}f^{=S} & \text{if } T=S\\
        0 & \text{else}
    \end{cases}
    \]
    Thus we can write
    \begin{align*}
    \norm{\tilde{g}|_x}_{2}^{2} &= \sum\limits_{S \ni i}\left(\left(\frac{2}{5}\right)^{|S|}f^{=S}+\Gamma^{(1)}_S\right)(x)^2 + \sum\limits_{S \subseteq [d]}\Gamma^{(2)}_S(x)^2\\
    &=\sum\limits_{S \ni i}\left(\left(\frac{2}{5}\right)^{2|S|}f^{=S}(x)^2+\Gamma^{(1)}_S(x)^2 + \frac{4}{5}f^{=S}(x)\Gamma^{(1)}_S(x)\right) + \sum\limits_{S \subseteq [d]}\Gamma^{(2)}_S(x)^2\\
    &\leq \tau + \sum\limits_{S \ni i}\left(\Gamma^{(1)}_S(x)^2 + 2^{O(d)}\Gamma^{(1)}_S(x)\right) + \sum\limits_{S \subseteq [d]}\Gamma^{(2)}_S(x)^2
    \end{align*}
    where all $\norm{\Gamma^{(i)}_S}_2 \leq 2^{O(d)}\gamma$ and we've used the facts that $i \notin J_x'$ and $\max_x\{f^{=S}\} \leq 2^d$. Finally, one can check that all $\Gamma^{(i)}_S$ are also bounded by $2^{O(d)}$ in infinity norm (this follows due to being differences of sums of Efron-Stein components, which are themselves bounded). Thus we can bound each square term by $2^{O(d)}$ times its linear component, and sum the linear components to get a vector $\Gamma(x)$ of $2$-norm at most $2^{O(d)}\gamma$. Taking the absolute value of this vector (which does not change its norm and continues to upper bound $\norm{\tilde{g}|_x}_{2}^{2} - \tau$) gives the desired expression.
\end{proof}
\section*{Acknowledgements}
We thank Vedat Alev, Venkat Guruswami, Tali Kaufman, and Shachar Lovett for helpful discussions on hypercontractivity and $q$-norm HDX. We thank Nima Anari especially for suggesting looking into interpolation type theorems for handling $q$-norms. We thank the Weizmann Institute and the Simons Institute for the Theory of Computing for hosting the author for part of the completion of this work.
\bibliographystyle{amsalpha}  
\bibliography{references} 
\end{document}